\documentclass[11pt]{article}
\usepackage{amsmath}
\usepackage{amsthm}
\usepackage{amssymb}
\usepackage{tabularx}
\usepackage[latin2]{inputenc}
\usepackage{graphicx}

\usepackage{sectsty}
\sectionfont{ \sc  \centering}
\subsectionfont{\sc \centering}
\subsubsectionfont{\sc \centering}
\paragraphfont{\sffamily \raggedright}
\allsectionsfont{\sffamily \centering}
\allsectionsfont{\bfseries \centering}

\newtheorem{theorem}{Theorem}[section]
\newtheorem{lemma}[theorem]{Lemma}
\newtheorem{claim}[theorem]{Claim}
\newtheorem{corollary}[theorem]{Corollary}
\newtheorem{proposition}[theorem]{Proposition}

\newcommand{\executeiffilenewer}[3]{%
\ifnum\pdfstrcmp{\pdffilemoddate{#1}}%
{\pdffilemoddate{#2}}>0%
{\immediate\write18{#3}}\fi%
} 

\newcommand{\svg}[2]{\def\svgwidth{#1}%
\executeiffilenewer{#2.svg}{#2.pdf}%
{inkscape -z -D --file=#2.svg %
--export-pdf=#2.pdf --export-latex}%
{\input{#2.pdf_tex}}}

\newcommand{\tw}{{\textbf{tw}}}
\newcommand{\pw}{{\textbf{pw}}}

\usepackage{todonotes}
\usepackage{vmargin}
\setmarginsrb{1.1in}{1.1in}{1.1in}{1.1in}{0mm}{0mm}{0mm}{7mm}

\newcommand{\ora}[1]{\overrightarrow{#1}}
\newcommand{\kkhs}{\textsc{$k\times k$ Hitting Set}}
\newcommand{\kkhsr}{\textsc{$k\times k$ Hitting Set$^{T}$}}
\newcommand{\tkkhsr}{\textsc{$2k\times 2k$ Hitting Set$^{T}$}}
\newcommand{\kkpermis}{\textsc{$k\times k$ Permutation Independent Set}}
\newcommand{\kkis}{\textsc{$k\times k$ Independent Set}}
\newcommand{\kkbis}{\textsc{$2k\times 2k$ Bipartite Permutation Independent Set}}
\newcommand{\kkcl}{\textsc{$k\times k$ Clique}}
\newcommand{\kkpermcl}{\textsc{$k\times k$ Permutation Clique}}
\newcommand{\dirpath}{\textsc{Directed Disjoint Paths}}
\newcommand{\disjpath}{\textsc{Disjoint Paths}}
\newcommand{\chnum}{\textsc{Chromatic Number}}
\newcommand{\dist}{\textsc{Distortion}}
\newcommand{\threecol}{\textsc{3-Coloring}}
\newcommand{\clstr}{\textsc{Closest String}}
\newcommand{\cperm}{\textsc{Constrained Permutation}}
\newcommand{\problem}[4]{
\begin{center}
\fbox{\parbox{0.9\linewidth}{
#1

\noindent\begin{tabularx}{\linewidth}{rX}
Input: & #2\\
Parameter: & #3\\
Question: & #4
\end{tabularx}}}
\end{center}
}

\date{}

\begin{document}

\title{Slightly Superexponential Parameterized Problems\footnote{A preliminary version of the paper appeared in the proceedings of SODA 2011.}}
\author{Daniel Lokshtanov\thanks{Department of Informatics, University of Bergen, Bergen, Norway. 
    \texttt{daniello@ii.uib.no}. Supported by ERC Starting Grant PaPaAlg (No. 715744).}
 \and D\'aniel Marx\thanks{Institute for Computer Science and Control, Hungarian Academy of Sciences (MTA SZTAKI), Budapest,
Hungary. \texttt{dmarx@cs.bme.hu}. Supported by ERC Starting Grant PARAMTIGHT (No. 280152) and Consolidator Grant SYSTEMATICGRAPH (No. 755978).}
 \and Saket Saurabh\thanks{The Institute of Mathematical Sciences, Chennai, India.    \texttt{saket@imsc.res.in}. Supported by the ERC Starting Grant PARAPPROX (No. 306992).}
}
\maketitle

\thispagestyle{empty}
\begin{abstract}
A central problem in parameterized algorithms is to obtain algorithms
 with running time $f(k)\cdot n^{O(1)}$ such that $f$ is as slow growing function of the parameter $k$ as possible.
In particular, a large number of basic parameterized problems admit parameterized algorithms where $f(k)$ is single-exponential, that is, $c^k$ for some constant $c$, which makes aiming for such a running time a natural goal for other problems as well. 
However there are still plenty of problems where the $f(k)$ appearing
in the best known running time is worse than single-exponential and it
remained ``slightly superexponential'' even after serious attempts to
bring it down.  A natural question to ask is whether the $f(k)$
appearing in the running time of the best-known algorithms is optimal
for any of these problems.

In this paper, we examine parameterized problems where $f(k)$ is
$k^{O(k)}=2^{O(k\log k)}$ in the best known running time and for a number of 
such problems, we show that the dependence
on $k$ in the running time cannot be improved to single exponential. 
More precisely we prove following tight lower bounds, for four natural problems, arising 
from three different domains:
\begin{itemize}
\item In the  \clstr\ problem, given strings $s_1$, $\dots$, $s_t$ over an alphabet
  $\Sigma$ of length $L$ each, and an integer $d$, the question is whether  there exists a string $s$ over $\Sigma$ of length $L$, such that its hamming distance from each of the strings $s_i$, $1\leq i \leq t$, is at most $d$. The pattern matching problem \clstr\  
is known to be solvable in
  time $2^{O(d\log d)}\cdot n^{O(1)}$ and $2^{O(d\log
    |\Sigma|)}\cdot n^{O(1)}$ . 
  We show that there are no $2^{o(d\log d)}\cdot n^{O(1)}$ or
  $2^{o(d\log |\Sigma|)}\cdot n^{O(1)}$ time algorithms, unless the Exponential Time
Hypothesis (ETH) fails.
\item The graph embedding problem \dist, that is, deciding whether a graph $G$ has a metric 
embedding into the integers with distortion at most $d$ can be solved in time $2^{O(d\log d)}\cdot
  n^{O(1)}$. We show that there is no $2^{o(d\log d)}\cdot n^{O(1)}$
  time algorithm, unless the ETH fails.
\item The \disjpath\ problem can be solved in time in time $2^{O(w\log
    w)}\cdot n^{O(1)}$ on graphs of treewidth at most $w$. We show that there is no $2^{o(w\log w)}\cdot n^{O(1)}$
  time algorithm, unless the ETH fails.
\item The \chnum\ problem can be solved in time in time $2^{O(w\log
    w)}\cdot n^{O(1)}$ on graphs of treewidth at most $w$. We show that there is no $2^{o(w\log w)}\cdot n^{O(1)}$
  time algorithm, unless the ETH fails.
\end{itemize}
To obtain our results, we first prove the lower bound for variants
of basic problems: finding cliques, independent sets, and hitting
sets. These artificially constrained variants  
form a good starting point for proving lower bounds on natural
problems without any technical restrictions and could be of independent interest. Several follow up works have already obtained tight lower bounds by using our framework, and we believe it will prove useful in obtaining even more lower bounds in the future. 
\end{abstract}
\pagestyle{plain}
\section{Introduction}
The goal of parameterized complexity is to find ways of solving
NP-hard problems more efficiently than brute force: our aim is to
restrict the combinatorial explosion to a parameter that is hopefully
much smaller than the input size. Formally, a {\em parameterization}
of a problem is assigning an integer $k$ to each input instance and we
say that a parameterized problem is {\em fixed-parameter tractable
  (FPT)} if there is an algorithm that solves the problem in time
$f(k)\cdot |I|^{O(1)}$, where $|I|$ is the size of the input and $f$ is an
arbitrary computable function depending on the parameter $k$
only. There is a long list of NP-hard problems that are FPT under
various parameterizations: finding a vertex cover of size $k$, finding
a cycle of length $k$, finding a maximum independent set in a graph
of treewidth at most $k$, etc.
For more background, the reader is referred to the monographs \cite{DBLP:books/sp/CyganFKLMPPS15,MR2001b:68042,grohe-flum-param,MR2223196}.

The practical applicability of fixed-parameter tractability results
depends very much on the form of the function $f(k)$ in the running
time. In some cases, for example in results obtained from Graph Minors
theory, the function $f(k)$ is truly horrendous (towers of
exponentials), making the result purely of theoretical interest. On
the other hand, in many cases $f(k)$ is a moderately growing
exponential function: for example, $f(k)$ is $1.2738^k$ in the current fastest
algorithm for finding a vertex cover of size $k$ \cite{DBLP:conf/mfcs/ChenKX06}, which can be further improved
to $1.1616^k$ in the special case of graphs with maximum degree 3 \cite{DBLP:conf/csr/Xiao08}. For
some problems, $f(k)$ can be even subexponential (e.g.,
$c^{\sqrt{k}}$) \cite{1070514,BidimensionalSurvey_GD2004,HMinorFree_JACM,DBLP:conf/icalp/AlonLS09}.

The implicit assumption in the research on fixed-parameter
tractability is that whenever a reasonably natural problem turns out
to be FPT, then we can improve $f(k)$ to $c^k$ with some small $c$
(hopefully $c<2$) if we work on the problem hard enough. Indeed, for
some basic problems, the current best running time was obtained after
a long sequence of incremental improvements. However, it is very well
possible that for some problems there is no algorithm with
single-exponential $f(k)$ in the running time.

In this paper, we examine parameterized problems where $f(k)$ is
``slightly superexponential'' in the best known running time: $f(k)$
is of the form $k^{O(k)}=2^{O(k\log k)}$.  Algorithms with this
running time naturally occur when a search tree of height at most $k$
and branching factor at most $k$ is explored, or when all possible
permutations, partitions, or matchings of a $k$ element set are
enumerated. For a number of such problems, we show that the dependence
on $k$ in the running time cannot be improved to single exponential.
More precisely, we show that a $2^{o(k\log k)}\cdot |I|^{O(1)}$ time
algorithm for these problems would violate the Exponential Time
Hypothesis (ETH), which is a complexity-theoretic assumption that can be informally stated as saying that there is no $2^{o(n)}$-time
algorithm for $n$-variable 3SAT \cite{MR1894519}.

In the first part of the paper, we prove the lower bound for variants
of basic problems: finding cliques, independent sets, and hitting
sets. These variants are artificially constrained such that the search
space is of size $2^{O(k\log k)}$ and we prove that a $2^{o(k\log k)} \cdot |I|^{O(1)}$
time algorithm would violate the ETH. The results in this section
demonstrate that for some problems the natural $2^{O(k\log k)} \cdot |I|^{O(1)}$ upper
bound on the search space is actually a tight lower bound on the
running time. More importantly, the results on these basic problems
form a good starting point for proving lower bounds on natural
problems without any technical restrictions.

In the second part of the paper, we use our results on the basic
problems to prove tight lower bounds for four natural problems from three
different domains:
\begin{itemize}
\item 
 In the  \clstr\ problem, given strings $s_1$, $\dots$, $s_t$ over an alphabet
  $\Sigma$ of length $L$ each, and an integer $d$, the question is whether  there exists a string $s$ over $\Sigma$ of length $L$, such that its hamming distance from each of the strings $s_i$, $1\leq i \leq t$, is at most $d$.  The pattern matching problem \clstr\ is known to be solvable in
  time $2^{O(d\log d)}\cdot  |I|^{O(1)}$ \cite{MR1984615} and $2^{O(d\log
    |\Sigma|)}\cdot  |I|^{O(1)}$ \cite{DBLP:journals/siamcomp/MaS09}. 
  We show that there are no $2^{o(d\log d)}\cdot n^{O(1)}$ or
  $2^{o(d\log |\Sigma|)}\cdot n^{O(1)}$ time algorithms, unless the ETH fails.

\item The graph embedding problem \dist, that is, deciding whether a $n$ vertex graph $G$ has a metric 
embedding into the integers with distortion at most $d$ can be done in time $2^{O(d\log d)}\cdot
  n^{O(1)}$\cite{DBLP:conf/icalp/FellowsFLLRS09}. We show that there is no $2^{o(d\log d)}\cdot n^{O(1)}$
  time algorithm, unless the ETH fails.
\item The \disjpath\ problem can be solved  in time $2^{O(w\log
    w)}\cdot n^{O(1)}$ on $n$ vertex graphs of treewidth at most $w$~\cite{Scheffler94}. 
We show that there is no $2^{o(w\log w)}\cdot n^{O(1)}$
  time algorithm, unless the ETH fails.
\item The \chnum\ problem can be solved  in time $2^{O(w\log
    w)}\cdot n^{O(1)}$ on $n$ vertex graphs of treewidth at most $w$~\cite{DBLP:journals/dam/JansenS97}.
We show that there is no $2^{o(w\log w)}\cdot n^{O(1)}$
  time algorithm, unless the ETH fails.
\end{itemize}
We remark that the algorithm given in~\cite{Scheffler94} does not mention the running time for \disjpath\ 
as $2^{O(w\log w)}\cdot n^{O(1)}$ on graphs of bounded treewidth but a closer look reveals that 
it is indeed the case. We expect that many further results of this form can be obtained by
using the framework of the current paper. Thus parameterized problems
requiring ``slightly superexponential'' time $2^{O(k\log k)}\cdot |I|^{O(1)}$ is not a
shortcoming of algorithm design or pathological situations, but an
unavoidable feature of the landscape of parameterized complexity.

It is important to point out that it is a real possibility that some
$2^{O(k\log k)}\cdot |I|^{O(1)}$ time algorithm can be improved to
single-exponential dependence with some work. In fact, there are
examples of well-studied problems where the running time was ``stuck''
at $2^{O(k\log k)} \cdot |I|^{O(1)}$ for several years before some new
algorithmic idea arrived that made it possible to reduce the
dependence to $2^{O(k)}\cdot |I|^{O(1)}$:
\begin{itemize}
\item In 1985, Monien \cite{MR87a:05097} gave a $k!\cdot n^{O(1)}$
  time algorithm for finding a cycle of length $k$ in a graph on $n$
  vertices. Alon, Yuster, and Zwick \cite{MR1411787} introduced the
  color coding technique in 1995 and used it to show that a cycle of
  length $k$ can be found in time $2^{O(k)}\cdot n^{O(1)}$.
\item In 1995, Eppstein~\cite{DBLP:journals/jgaa/Eppstein99} gave a $O(k^k n)$
  time algorithm for deciding if a $k$-vertex planar graph $H$ is a
  subgraph of an $n$-vertex planar graph $G$. 
  Dorn~\cite{DBLP:conf/stacs/Dorn10} gave an improved algorithm with
  running time $2^{O(k)}\cdot n$. One of the main technical tools in
  this result is the use of sphere cut decompositions of planar
  graphs, which was used earlier to speed up algorithms on planar
  graphs in a similar way \cite{DBLP:conf/esa/DornPBF05}.
\item In 1995, Downey and Fellows \cite{downey-fellows-feasibility} gave a
  $k^{O(k)}\cdot n^{O(1)}$ time algorithm for \textsc{Feedback Vertex Set} (given an
  undirected graph $G$ on $n$ vertices, delete $k$ vertices to make it acyclic). A
  randomized $4^k\cdot n^{O(1)}$ time algorithm was given in 2000~\cite{DBLP:journals/jair/BeckerBG00}. The
  first deterministic $2^{O(k)}\cdot n^{O(1)}$ time algorithms appeared only in
  2005 \cite{GGHNW-2005,mr2352543}, using the technique of
  iterative compression introduced by Reed et al. \cite{ReedSmithVetta-OddCycle}.
\item In 2003, Cook and Seymour \cite{DBLP:journals/informs/CookS03} used standard dynamic programming techniques to give a $2^{O(w\log w)}$ $n^{O(1)}$-time algorithm for \textsc{Feedback Vertex Set} on graphs of treewidth $w$, and it was considered plausible that this is the best possible form of running time. Hence it was a remarkable surprise in 2011 when Cygan et al.~\cite{Cygan2011} presented a $3^{w}n^{O(1)}$ time randomized algorithm by using the so-called Cut \& Count technique. Later, Bodlaender et al.~\cite{BodlaenderCKN2015} and Fomin et al.~\cite{FedorDS2014} obtained deterministic single-exponential parameterized algorithms using a different approach.

\end{itemize}
As we can see in the examples above, achieving single-exponential
running time often requires the invention of significant new
techniques. Therefore, trying to improve the running time for a problem whose best known
parameterized algorithm is slightly superexponential can lead to
important new discoveries and developments. However, in this paper we
identify problems for which such an improvement is very unlikely. The
$2^{O(k\log k)}$ dependence on $f(k)$ seems to be inherent to these problems, or one should realize that in achieving 
single-exponential dependence one is essentially trying to disprove the ETH.

There are some lower bound results on FPT problems in the
parameterized complexity literature, but not of the form that we are
proving here. Cai and
Juedes~\cite{DBLP:journals/jcss/CaiJ03} proved that if the
parameterized version of a MAXSNP-complete
problems (such as \textsc{Vertex Cover} on graphs of maximum degree 3)
can be solved in time $2^{o(k)}\cdot |I|^{O(1)}$, then ETH fails. Using
parameterized reductions, this result can be transfered to other
problems: for example, assuming the ETH, there is a no $2^{o(\sqrt{k})}\cdot |I|^{O(1)}$ time
algorithm for planar versions of \textsc{Vertex Cover},
\textsc{Independent Set}, and \textsc{Dominating Set} (and this bound
is tight). However, no lower bound above $2^{O(k)}$ was obtained this
way for any problem so far.

Flum, Grohe, and Weyer \cite{DBLP:journals/jcss/FlumGW06} tried to
rebuild parameterized complexity by redefining fixed-parameter
tractability as $2^{O(k)}\cdot |I|^{O(1)}$ time and introducing
appropriate notions of reductions, completeness, and complexity
classes. This theory could be potentially used to show that the
problems treated in the current paper are hard for certain classes,
and therefore they are unlikely to have single-exponential
parameterized algorithms. However, we see no reason why these
problems would be complete for any of those classes (for example, the
only complete problem identified in \cite{DBLP:journals/jcss/FlumGW06}
that is actually FPT is a model checking on problem on words for which
it was already known that $f(k)$ cannot even be elementary).  Moreover, we are not
only giving evidence against single-exponential time algorithms in this paper,
but show that the $2^{O(k\log k)}$ dependence is actually tight.


\section{Basic problems}
In this section, we modify basic problems in such a way that they can
be solved in time $2^{O(k\log k)}|I|^{O(1)}$ by brute force, and this is best possible
assuming the ETH. In all the problems of this section, the task is to
select exactly one element from each row of a $k\times k$ table such
that the selected elements satisfy certain constraints. This means
that the search space is of size $k^k=2^{O(k\log k)}$. We denote by
$[k]\times [k]$ the set of elements in a $k\times k$ table, where
$(i,j)$ is the element in row $i$ and column $j$. Thus selecting
exactly one element from each row gives a set $(1,\rho(1))$, $\dots$,
$(k,\rho(k))$ for some mapping $\rho: [k]\to [k]$. In some of the
variants, we not only require that exactly one element is selected
from each row, but we also require that exactly one element is
selected from each column, that is, $\rho$ has to be a
permutation. The lower bounds for such permutation problems will
be essential for proving hardness results on \clstr\
(Section~\ref{sec:closest-substring}) and \textsc{Distortion}
(Section~\ref{sec:distortion}). The key step in obtaining the lower
bounds for permutation problems is the randomized reordering argument
of Theorem~\ref{th:kkbis}. The analysis and derandomization of this
step is reminiscent of the color coding \cite{MR1411787} and chromatic
coding \cite{DBLP:conf/icalp/AlonLS09} techniques. 

To prove that a too fast algorithm for a certain problem $P$
contradicts the Exponential Time Hypothesis, we have to reduce
$n$-variable 3SAT to problem $P$ and argue that the algorithm would
solve 3SAT in time $2^{o(n)}$. It will be somewhat more convenient to do
the reduction from \textsc{3-Coloring}. We use the well-known fact
that there is a polynomial-time reduction from 3SAT to
\textsc{3-Coloring} where the number of vertices of the graph is linear in the size formula.
\begin{proposition}\label{prop:sat-3col}
Given a 3SAT formula $\phi$ with $n$-variables and $m$-clauses, it is
possible to construct a graph $G$ with $O(n+m)$ vertices in polynomial
time such that $G$ is 3-colorable if and only if $\phi$ is
satisfiable.
\end{proposition}
Proposition~\ref{prop:sat-3col} implies that an algorithm for
\textsc{3-Coloring} with running time subexponential in the number of
vertices gives an algorithm for 3SAT that is subexponential in the
number of {\em clauses}. This is sufficient for our purposes, as the
Sparsification Lemma of Impagliazzo, Paturi and Zane \cite{MR1894519} shows that such
an algorithm already violates the ETH.
\begin{lemma}[\cite{MR1894519}]\label{lem:sparse}
Assuming the ETH, there is no $2^{o(m)}$ time algorithm for $m$-clause 3SAT.
\end{lemma}
Combining Proposition~\ref{prop:sat-3col} and Lemma~\ref{lem:sparse} gives
the following proposition:

\begin{proposition}\label{prop:3col}
Assuming the ETH, there is no $2^{o(n)}$ time algorithm for deciding whether 
an $n$-vertex graph is $3$-colorable.
\end{proposition}

\subsection{\kkcl}

The first problem we investigate is the variant of the standard clique
problem where the vertices are the elements of a $k\times k$ table,
and the clique we are looking for has to contain exactly one element
from each row.

\problem{\kkcl}{A graph $G$ over the vertex set $[k]\times [k]$}
{$k$}{Is there a $k$-clique in $G$ with exactly one
element from each row?}
Note that the graph $G$ in the \kkcl\ instance has $O(k^2)$ vertices
at most $O(k^4)$ edges, thus the
size of the instance is $O(k^4)$.

\begin{theorem}\label{th:kkcl}
Assuming the ETH, there is no $2^{o(k\log k)}$ time algorithm for \kkcl.
\end{theorem}
\begin{proof}
Suppose that there is an algorithm $\mathbb{A}$ that solves \kkcl\ in
$2^{o(k\log k)}$ time. We show that this implies that 
\threecol\ on a graph with $n$ vertices can be solved in time
$2^{o(n)}$, which contradicts the ETH by Proposition~\ref{prop:3col}.

Let $H$ be a graph with $n$ vertices. Let $k$ be the smallest integer
such that $3^{n/k+1}\le k$, or equivalently, $n\le k\log_3 k-k$. Note that
such a finite $k$ exists for every $n$ and  it is easy
to see that $k\log k=O(n)$ for the smallest such $k$. Intuitively, it
will be useful to think of $k$ as a value somewhat larger than $n/\log
n$ (and hence $n/k$ is somewhat less than $\log n$).

Let us partition the vertices of $H$ into $k$ groups $X_1$, $\dots$,
$X_k$, each of size at most $\lceil n/k \rceil$. For every $1\le i \le
k$, let us fix an enumeration of all the proper 3-colorings of
$H[X_i]$. Note that there are most $3^{\lceil n/k\rceil}\le 3^{n/k+1}
  \le k$ such 3-colorings for every $i$. We say that a proper 3-coloring
  $c_i$ of $H[X_i]$ and a proper 3-coloring $c_j$ of $H[X_j]$ are {\em
    compatible} if together they form a proper coloring of $H[X_i\cup
  X_j]$: for every edge $uv$ with $u\in X_i$ and $v\in X_j$, we have
  $c_i(u)\neq c_j(v)$.  Let us construct a graph $G$ over the vertex set
  $[k]\times [k]$ where vertices $(i_1,j_1)$ and $(i_2,j_2)$ with
  $i_1\ne i_2$ are adjacent if and only if the $j_1$-th proper
  coloring of $H[X_{i_1}]$ and the $j_2$-th proper coloring of
  $H[X_{i_2}]$ are compatible (this means that if, say, $H[X_{i_1}]$
  has less than $j_1$ proper colorings, then $(i_1,j_1)$ is an
  isolated vertex).

We claim that $G$ has a $k$-clique having exactly one vertex from each
row if and only if $H$ is 3-colorable. Indeed, a proper 3-coloring of $H$
induces a proper 3-coloring for each of $H[X_1]$, $\dots$,
$H[X_k]$. Let us select vertex $(i,j)$ if and only if the proper
coloring of $H[X_i]$ induced by $c$ is the $j$-th proper coloring of
$H[X_i]$. It is clear that we select exactly one vertex from each row
and they form a clique: the proper colorings of $H[X_i]$ and $H[X_j]$
induced by $c$ are clearly compatible. For the other direction,
suppose that $(1,\rho(1))$, $\dots$, $(k,\rho(k))$ form a $k$-clique
for some mapping $\rho:[k]\to[k]$. Let $c_i$ be the $\rho(i)$-th
proper 3-coloring of $H[X_i]$. The colorings $c_1$, $\dots$, $c_k$
together define a coloring $c$ of $H$. This coloring $c$ is a proper
3-coloring: for every edge $uv$ with $u\in X_{i_1}$ and $v\in X_{i_2}$, the
fact that $(i_1,\rho(i_1))$ and $(i_2,\rho(i_2))$ are adjacent means
that $c_{i_1}$ and $c_{i_2}$ are compatible, and hence $c_{i_1}(u)\neq c_{i_2}(v)$.

Running the assumed algorithm $\mathbb{A}$ on $G$ decides the
3-colorability of $H$. Let us estimate the running time of
constructing $G$ and running algorithm $\mathbb{A}$ on $G$. The graph
$G$ has $k^2$ vertices and the time required to construct $G$ is
polynomial in $k$: for each $X_i$, we need to enumerate at most $k$
proper 3-colorings of $G[X_i]$. Therefore, the total running time
is $2^{o(k\log k)}\cdot k^{O(1)}=2^{o(n)}$ (using that $k\log k=O(n)$). It follows that we have a
$2^{o(n)}$ time algorithm for \threecol\ on an $n$-vertex graph, contradicting the ETH.
\end{proof}

\kkpermcl\ is a more restricted version of \kkcl: in addition to requiring that the
clique contains exactly one vertex from each {\em row,} we also require that it
contains exactly one vertex from each {\em column.} In other words, the
vertices selected in the solution are $(1,\rho(1))$, $\dots$, $(k,\rho(k))$ for some
{\em permutation} $\rho$ of $[k]$. Given an instance $I$ of \kkcl\
having a solution $S$, if
we randomly reorder the vertices in each row, then with some
probability the reordered version of solution $S$ contains exactly one
vertex from each row and each column of the reordered instance. In
Theorem~\ref{th:kkpermcl}, we use this argument to show that a
$2^{o(k\log k)}$ time algorithm for \kkpermcl\ gives a {\em randomized}
$2^{o(k\log k)}$ time algorithm for \kkcl. Section~\ref{sec:derandomization} shows how the proof
of Theorem~\ref{th:kkpermcl} can be derandomized.

\begin{theorem}\label{th:kkpermcl}
If there is a $2^{o(k\log k)}$ time algorithm for \kkpermcl, then
there is a randomized $2^{o(m)}$ time algorithm  for $m$-clause 3SAT.
\end{theorem}
\begin{proof}
We show how to transform an instance $I$ of \kkcl\ into an instance $I'$ of
\kkpermcl\ with the following properties: if $I$ is a no-instance,
then $I'$ is a no-instance, and if $I$ is a yes-instance, then $I'$ is
a yes-instance with probability at least $2^{-O(k)}$.  
This means that if we perform this transformation $2^{O(k)}$ times and
accept $I$ as a yes-instance if and only at least one of the
$2^{O(k)}$ constructed instances is a yes-instance, then the
probability of incorrectly rejecting a yes-instance can be reduced to
an arbitrary small constant. Therefore, a $2^{o(k\log k)}$ time
algorithm for \kkpermcl\ implies a randomized $2^{O(k)}\cdot
2^{o(k\log k)}=2^{o(k\log k)}$ time
algorithm for \kkcl.

Let $c(i,j):[k]\times[k]\to [k]$ be a mapping chosen uniform at
random; we can imagine $c$ as a coloring of the $k\times k$ vertices.
Let $c'(i,j)=\bigstar$ if there is a $j'\neq j$ such that
$c(i,j)=c(i,j')$ and let $c'(i,j)=c(i,j)$ otherwise (i.e., if
$c(i,j)=x\neq \bigstar$, then no other vertex has color $x$ in row
$i$). The instance $I'$ of \kkpermcl\ is constructed the following
way: if there is an edge between $(i_1,j_1)$ and $(i_2,j_2)$ in
instance $I$ and $c'(i_1,j_1),c'(i_2,j_2)\neq \bigstar$, then we add an
edge between $(i_1,c'(i_1,j_1))$ and $(i_2,c'(i_2,j_2))$ in instance
$I'$. That is, we use mapping $c$ to rearrange the vertices in each
row. If vertex $(i,j)$ clashes with some other vertex in the same row
(that is, $c(i,j)=\bigstar$), then all the edges incident to $(i,j)$ are
thrown away.

Suppose that $I'$ has a $k$-clique $(1,\rho(1))$, $\dots$,
$(k,\rho(k))$ for some permutation $\rho$ of $[k]$. For every $i$,
there is a unique $\delta(i)$ such that $c'(i,\delta(i))=\rho(i)$:
otherwise $(i,\rho(i))$ is an isolated vertex in $I'$. It is easy to
see that $(1,\delta(i))$, $\dots$, $(k,\delta(k))$ is a clique in $I$:
vertices $(i_1,\delta(i_1))$ and $(i_2,\delta(i_2))$ have to be
adjacent, otherwise there would be no edge between $(i_1,\rho(i_1))$
and $(i_2,\rho(i_2))$ in $I'$. Therefore, if $I$ is a no-instance,
then $I'$ is a no-instance as well.

Suppose now that $I$ is a yes-instance: there  is a clique
$(1,\delta(1))$, $\dots$, $(k,\delta(k))$ in $I$. Let us estimate the
probability that the following two events occur:
\begin{itemize}
\setlength\itemsep{-.9mm}
\item[(1)] For every $1\le i_1<i_2\le k$, $c(i_1,\delta(i_1))\neq
  c(i_2,\delta(i_2))$.
\item[(2)] For every $1\le i \le k$ and $1\le j \le k$ with $j\neq
  \delta(i)$, $c(i,\delta(i))\neq c(i,j)$.
\end{itemize}
Event (1) means that $c(1,\delta(1))$, $\dots$, $c(k,\delta(k))$ is a
permutation of $[k]$. Therefore, the probability of (1) is
$k!/k^k=e^{-O(k)}$ (using Stirling's Formula). For a particular $i$,
event (2) holds if $k-1$ randomly chosen values are all different from $c(i,\delta(i))$. Thus
the probability that (2) holds for a particular $i$ is
$(1-1/k)^{-(k-1)}\ge e^{-1}$ and the probability that (2) holds for
every $i$ is at least $e^{-k}$.  Furthermore, events (1) and (2) are
independent: we can imagine the random choice of the mapping $c$ as
first choosing the values $c(1,\delta(1))$, $\dots$, $c(k,\delta(k))$
and then choosing the remaining $k^2-k$ values. Event (1) depends only
on the first $k$ choices, and for any fixed result of the first $k$
choices, the probability of event (2) is the same. Therefore, the
probability that (1) and (2) both hold is $e^{-O(k)}$. 

Suppose that (1) and (2) both hold. Event (2) implies that
$c(i,\delta(i))=c'(i,\delta(i))
\neq \bigstar$ for every $1\le i \le k$. Event (1) implies that if we set
$\rho(i):=c(i,\delta(i))$, then $\rho$ is a permutation of
$[k]$. Therefore, the clique $(1,\rho(1))$, $\dots$, $(k,\rho(k))$ is a
solution of $I'$, as required.
\end{proof}
In the next section, we show that instead of random colorings, we can
use a certain deterministic family of colorings. This will imply:
\begin{corollary}\label{cor:detpermcl}
Assuming the ETH, there is no $2^{o(k\log k)}$ time algorithm for
\kkpermcl.
\end{corollary}

\subsubsection{Derandomization}\label{sec:derandomization}
In this section, we give a coloring family that can be used instead of
the random coloring in the proof of Theorem~\ref{th:kkpermcl}.  We
call a graph $G$ to be a {\em cactus-grid graph} if the vertices are
elements of a $k\times k$ table and the graph precisely consists of a
clique containing exactly one vertex from each row and each 
vertex in the clique is adjacent to every other vertex in its row.
There are no other edges in the graph, thus the graph has exactly
$\binom{k}{2}+k(k-1)$ edges.
We are interested in  a coloring family ${\cal
  F}=\left\{f:[k]\times [k]\rightarrow [k+1] \right\}$ with the
property that for any cactus-grid graph $G$ with vertices from
$k\times k$ table, there exists a function $f\in {\cal F}$ such that
$f$ properly colors the vertices of $G$. We call such a $\cal F$ as a
coloring family for cactus-grid graphs.

Before we proceed to construct a coloring family $\cal F$ of size
$2^{O(k \log \log k)}$, we explain how this can be used to obtain the derandomized
version of Theorem~\ref{th:kkpermcl}, the
Corollary~\ref{cor:detpermcl}.  Suppose that the instance $I$ of
\kkcl\ is a yes-instance. Then there is a clique $(1,\delta(1))$,
$\dots$, $(k,\delta(k))$ in $I$. Consider the cactus-grid graph $G$
consisting of clique $(1,\delta(1))$, $\dots$, $(k,\delta(k))$ and for
each $1\le i\le k$, the edges between $(i,\delta(i))$ and $(i,j)$ for
every $j\neq \delta(i)$. Let $f\in {\cal F}$ be a proper coloring of
$G$. Now since $(1,\delta(1))$, $\dots$, $(k,\delta(k))$ is a clique
in $G$ they get  distinct colors by $f$ and since all the vertices in
the row $i$, $(i,j)$, $j\neq \delta(i)$, are adjacent to
$(i,\delta(i))$ we have that $f((i,j))\neq f(i,\delta(i))$.  So if we
use this $f$ in place of $c(i,j)$, the random coloring used in the
proof of Theorem~\ref{th:kkpermcl}, then events (1) and (2) hold and
we know that the instance $I'$ obtained using $f$ is a yes-instance of
\kkpermcl. Thus we know that an instance $I$ of \kkcl\ has a clique of
size $k$ containing exactly one element from each row if and only if
there exists an $f\in {\cal F}$ such that the corresponding instance
$I'$ of \kkpermcl\ has a clique of size $k$ such that it contains
exactly one element from each row and column. This together with the
fact that the size of ${\cal F}$ is bounded by $2^{O(k \log \log k)}$ imply the
Corollary~\ref{cor:detpermcl}. 

To construct our deterministic coloring family we also 
 need a few known results on perfect hash families. Let ${\cal
  H}=\{f:[n]\rightarrow [k] \} $ be a set of functions such that for
all subsets $S$ of size $k$ there is a $h\in {\cal H}$ such that it is
one-to-one on $S$. The set $\cal H$ is called {\em $(n,k)$-family of
  perfect hash functions}. There are some known constructions for set
$\cal H$. We summarize them below.
\begin{proposition}[\cite{MR1411787,NaorSS95}]
There exists explicit construction $\cal H$ of $(n,k)$-family of perfect hash functions of size $O(11^k \log n)$. 
There is also another explicit construction $\cal H$ of $(n,k)$-family of perfect hash functions of size 
$O(e^k k^{O(\log k)} \log n)$.  
\end{proposition}

Now we are ready to state the main lemma of this section. 

\begin{lemma}
\label{lem:cactuscolfam}
There exists explicit construction of coloring family $\cal F$  for cactus-grid graphs of size $2^{O(k \log \log k)}$. 
 \end{lemma}

\begin{proof}
Our idea for deterministic coloring family  $\cal F$ for cactus-grid graphs is to keep $k$ functions 
$f_1,\ldots, f_k$ where each $f_i$ is an element of a $(k,k')$-family of perfect hash functions for 
some $k'$ and use it to map the elements of $\{i\}\times k$ (the column $i$). We guess the number of vertices of $G$ that appear in each column, and we reserve that many private colors for the column so that these colors are not used on the vertices of any other columns. 
This will ensure that we get the desired coloring family.  
We make our intuitive idea more precise below. A description of a function $f\in {\cal F}$ consists of a tuple having 
\begin{itemize}
\item a set $S\subseteq [k]$;
\item a tuple $(k_1,k_2,\ldots,k_\ell)$ where $k_i\geq 1$, $\ell=|S|$ and $\sum_{i=1}^\ell k_i=k$;
\item $\ell$ functions $f_1,\ldots,f_\ell$ where $f_i\in {\cal H}_i$ and ${\cal H}_i$ is a $(k,k_i)$-family of 
      perfect hash functions. 
\end{itemize}
The set $S$ tells us which columns the clique intersects. Let the elements of $S=\{s_1,\ldots,s_\ell\}$ be sorted 
in increasing order, say $s_1<s_2 < \dots <s_{\ell}$. Then the tuple  $(k_1,k_2,\ldots,k_\ell)$ tells 
us that the column $s_j$, $1\leq j \leq \ell$, contains $k_j$ vertices from the clique. Hence 
with this interpretation, given a tuple $(S,(k_1,\ldots,k_\ell),f_1,\ldots,f_\ell)$ we define the coloring 
function $g:[k]\times [k] \rightarrow [k]$ as follows. Every element in $[k]\times \{1,\ldots,k\}\setminus S$ is mapped 
to $k+1$. Now for vertices in $[k] \times \{s_j\}$ (vertices in column $s_j$), we define $g(i,s_j)=f_j(i)+\sum_{1\leq i<j} k_i$. 
We do this for every $j$ between $1$ and $\ell$. This concludes the description. Now we show that it is indeed a coloring family 
for cactus-grid graphs. Given a cactus grid graph $G$, we first look at the columns it intersects and that forms our set $S$ and 
then the number of vertices it intersects in each column makes the the tuple $(k_1,k_2,\ldots,k_\ell)$. Finally for each of 
the columns there exists a function $h$ in the perfect $(k,k_i)$-hash family that maps the 
elements of clique in this column one to one with $[k_i]$; we store this function corresponding to this column. Now we show that the function $g$ corresponding to this tuple properly 
colors $G$. The function $g$ assigns different values from $[k]$ to the columns in $S$ and hence 
we have that the vertices of clique gets distinct colors as in each column we have  
a function $f_i$ that is one-to-one on the vertices of $S$. Now we look at the 
edge with both end-points in the same row. If any of the end-point occurs in column that is not in $S$, then we know that it has been 
assigned $k+1$ while the vertex from the clique has been assigned color from $[k]$. If both end-points are from $S$, then the 
offset we use to give different colors to vertices in these columns ensures that these end-points get different colors. This shows that 
$g$ is indeed a proper coloring of $G$. 
This shows that for every cactus-grid graph we have a function 
$g\in \cal F$. Finally, the bound on the size of $\cal F$ is as follows, 
\begin{eqnarray}\label{eqn:firstcolbd}
 2^k 4^k \prod_{i=1}^\ell (11^{k_i}\log k)\leq 2^{O(k)}(\log k)^\ell \leq 2^{O(k \log \log k)}. 
\end{eqnarray}
This concludes the proof. 
\end{proof}

 
The bound achieved in Equation~\ref{eqn:firstcolbd} 
on the size of $\cal F$ is sufficient for our purpose but it is not as small as $2^{O(k)}$ 
that one can obtain using a simple application of probabilistic methods. We provide a family $\cal F$ of size $2^{O(k)}$ below  
which could be of independent algorithmic interest.

\begin{lemma}\label{lem:cactuscolfamnew}
There exists explicit construction of coloring family $\cal F$  for cactus-grid graphs of size $2^{O(k)}$. 
 \end{lemma}
\begin{proof}
We incurred a factor of $(\log k)^\ell$ in the construction given in Lemma~\ref{lem:cactuscolfam} because for every column we applied hash functions from 
$[k]\rightarrow [k_i]$. Loosely speaking, if we could replace these by $[k_i^2]\rightarrow [k_i]$, then the size of family 
will be $11^{k_i} \log k_i \leq 12^{k_i}$ and then $ \prod_{i=1}^\ell 11^{k_i}\log k_i\leq 12^k$. 
Next we describe 
a procedure to do this by incurring an extra  cost of 
 $2^{O(\log ^3 k)}$. To do this we use the following classical lemma proved by 
Fredman, Koml\'os and Szemer\'edi~\cite{FredmanKS84}.
\begin{lemma}[\cite{FredmanKS84}]
\label{lemma:fks}
Let $W\subseteq[n]$ with $|W|=r$. The mapping $f:[n]\rightarrow [2r^2]$ such that $f(x)=(t x \mod p)\mod~2r^2$ is one-to-one 
when restricted to $W$ for at least half of the values $t\in [p]$.  Here $p$ is any prime between $n$ and $2n$.  
\end{lemma}
The idea is to use Lemma~\ref{lemma:fks} to choose multipliers ($t$ in the above description) appropriately. Let us fix a prime $p$ between 
$k$ and $2k$.   
Given a set $S$ and a tuple $(k_1,k_2,\ldots,k_\ell)$ we make a partition of set $S$ as follows  
$S_{i}=\{s_j ~|s_j\in S, 2^{i-1}< k_j\leq 2^i\}$ for $i \in \{0,\ldots, \lceil \log k \rceil\}$. 
Now let us fix a set $S_i$, by our construction we know that the size of intersection of the clique with each of the columns in $S_i$  
is roughly same. For simplicity of argument, let us fix a clique $W$ of some cactus grid graph $G$. 
Consider a bipartite graph $(A,B)$ where $A$ contains a vertex for each column in $S_i$ and $B$ consists of 
numbers from $[p]$. Now we give an edge between vertex $a\in A$ and $b\in B$ if we can use $b$ as a multiplier in 
Lemma~\ref{lemma:fks}, that is, the map $f(x)=(bx~ \mod~p)\mod~2^{2i+1}$ 
is one-to-one when restricted to the vertices of the clique $W$ to the column $a$. 

Observe that because of Lemma~\ref{lemma:fks}, every vertex 
in $A$ has degree at least $p/2$ and hence there exists a vertex $b \in B$ that can be used as a multiplier 
for at least half of the elements in the set $A$. We can repeat this argument by removing a vertex 
$b\in B$, that could be used as a multiplier for half of the vertices in $A$,  
and all the columns for which it can be multiplier.  This implies that there exits a set $X_i\subseteq [p]$ of size 
$\log |A|\leq \log k$ that could be used as a multiplier for every column in $A$.  
Now we give a description of a function $f\in {\cal F}$ that consists of a tuple having 
\begin{itemize}

\item a set $S\subseteq [k]$;
\item a tuple $(k_1,k_2,\ldots,k_\ell)$ where $k_i\geq 1$, $\ell=|S|$ and $\sum_{i=1}^\ell k_i=k$;
\item $((b^i_1,\ldots, b^{i}_q), (L^i_1,\ldots,L^i_q))$, $1\leq i \leq \lceil \log k \rceil $, $q=\lceil \log k \rceil$; 
Here $(L^i_1,\ldots,L^i_q)$ is a partition of $S_i$ and the interpretation is that for every column in 
$L_j^i$ we will use $b^i_j$ as a multiplier for range reduction; 
\item $\ell$ functions $f_1,\ldots,f_\ell$ where $f_i\in {\cal H}_i$ and ${\cal H}_i$ is a $(8k_i^2,k_i)$-family of 
      perfect hash functions. 
\end{itemize}
This completes the description. Now given a tuple 
$$(S,(k_1,\ldots,k_\ell),\{((b^i_1,\ldots, b^{i}_q), (L^i_1,\ldots,L^i_q))~|~ 1\leq i \leq \lceil \log k \rceil   \},  f_1,\ldots,f_\ell)$$ 
we define the coloring function $g:[k]\times [k] \rightarrow [k]$ as follows. Every element in $[k]\times \{1,\ldots,k\}\setminus S$ is mapped 
to $k+1$. Now for vertices in $[k] \times \{s_j\}$ (vertices in column $s_j$), we do as follows. 
Suppose $s_j\in L_\alpha^{\beta}$ then we define $g(i,s_j)=(\sum_{1\leq i<j} k_i)+f_j(((b_{\alpha}^{\beta} s_j)~\mod~p)\mod~ck_j^2 )$. 
We do this for every $j$ between $1$ and $\ell$. This concludes the description for $g$. Observe that given a vertex in column $s_j$ 
we first use the function in Lemma~\ref{lemma:fks} to reduce its range to roughly $O(k_j^2)$ and still preserving that for every subset $[k]$ of 
size at most $2k_j$ there is some multiplier which maps it injective. It is evident from the above description that this is indeed a 
coloring family of cactus grid graphs. The range of any function in $\cal F$ is $k+1$ and the size of this family is 
\[2^k 4^k  \prod_{i=1}^{\lceil \log k \rceil}(p)^{\log k}  \prod_{i=1}^{\lceil \log k \rceil}4^{\sum_{j=1}^{\lceil \log k \rceil} |L_j^i|}   \prod_{i=1}^\ell (11^{k_i}\log k_i) \leq 8^k (2k)^{\log k} 4^k 12^k \leq 2^{O(k+(\log
  k)^3)}\leq 2^{O(k)}. \]
The last assertion follows from the fact that $\sum_{i=1}^{\lceil \log k \rceil}\sum_{j=1}^{\lceil \log k \rceil} |L_j^i| \leq k$ and $\sum_{i=1}^\ell k_i=k$. This concludes the proof. 
\end{proof}

\subsection{\kkis}\label{sec:kkis}
The lower bounds in Section~\ref{th:kkcl} for \textsc{$k\times k$
  (Permutation) Clique} obviously hold
for the analogous \textsc{$k\times k$
  (Permutation) Independent Set} problem: by taking the complement of the graph, we can
reduce one problem to the other. We state here a version of the
independent set problem
that will be a convenient starting point for reductions in later sections:

\problem{\kkbis}{A graph $G$ over the vertex set $[2k]\times [2k]$
  where every edge is between $I_1=\{(i,j)\mid i,j\le k\}$ and
  $I_2=\{(i,j)\mid i,j\ge k+1\}$.}
{$k$}{Is there an independent set $(1,\rho(1))$, $\dots$,
  $(2k,\rho(2k)) \subseteq I_1\cup I_2$ in $G$ for some permutation $\rho$ of $[2k]$?}
That is, the upper left quadrant $I_1$ and the lower right quadrant
$I_2$ induce independent sets, and every edge is between these two
independent sets. The requirement that the solution is a subset of $I_1\cup I_2$ means
that $\rho(i)\le k$ for $1\le i \le k$ and $\rho(i)\ge k+1$ for $k+1
\le i \le 2k$.
\begin{theorem}
\label{th:kkbis}
Assuming the ETH, there is no $2^{o(k\log k)}$ time algorithm for \kkbis.
\end{theorem}
\begin{proof}
  Given an instance $I$ of \kkpermis, we construct an equivalent instance
  $I'$ of \kkbis\ the following way. For every $1\le i \le k$ and
  $1\le j,j ' \le k$, $j\neq j'$, we add an edge between $(i,j)$ and
  $(i+k,j'+k)$ in $I'$. If there is an edge between $(i_1,j_1)$ and
  $(i_2,j_2)$ in $I$, then we add an edge between $(i_1,j_1)$ and
  $(i_2+k,j_2+k)$ in $I'$. This completes the description of $I'$.

  Suppose that $I$ has a solution $(1,\delta(1))$, $\dots$,
  $(k,\delta(k))$ for some permutation $\delta$ of $[2k]$. Then it is
  obvious from the construction of $I'$ that $(1,\delta(1))$, $\dots$,
  $(k,\delta(k))$, $(1+k,\delta(1)+k)$, $\dots$, $(2k,\delta(k)+k)$ is
  an independent set of $I'$ and $\delta(1)$, $\dots$, $\delta(k)$,
  $\delta(1)+k$, $\dots$, $\delta(k)+k$ is clearly a permutation of
  $[2k]$. Suppose that $(1,\rho(1))$, $\dots$, $(2k,\rho(2k))$ is
  solution of $I'$ for some permutation $\rho$ of $[2k]$. By
  definition, $\rho(i)\le k$ for $1\le i\le k$. We claim that
  $(1,\rho(k))$, $\dots$, $(k,\rho(k))$ is an independent set of $I$.
  Observe first that $\rho(i+k)=\rho(i)+k$ for every $1\le i \le k$:
  otherwise there is an edge between $(i,\rho(i))$ and
  $(i+k,\rho(i+k))$ in $I'$. If there is an edge between
  $(i_1,\rho(i_1))$ and $(i_2,\rho(i_2))$ in $I$, then by construction
  there is an edge between $(i_1,\rho(i_1))$ and
  $(i_2+k,\rho(i_2)+k)=(i_2+k,\rho(i_2+k))$ in $I'$, contradicting the
  assumption that $(1,\rho(k))$, $\dots$, $(2k,\rho(2k))$ is an
  independent set in $I'$.
\end{proof}

%

\subsection{\kkhs}
\textsc{Hitting Set} is a W[2]-complete problem, but if we restrict
the universe to a $k\times k$ table where only one element can be
selected from each row, then it can be solved in time $O^*(k^k)$ by
brute force.

\problem{\kkhs}{Sets $S_1,\dots,S_m\subseteq [k]\times [k]$.}
{$k$}
{Is there a set $S$ containing exactly one element from each row such
  that $S\cap S_i\neq \emptyset$ for any $1\le i \le m$?}
We say that the mapping $\rho$ {\em hits}
a set $S\subseteq [k]\times [k]$, if $(i,\rho(i))\in m$ for some $1\le
i \le S$. Note that unlike for \kkcl\ and \kkis, the size of the
\kkhs\ instance cannot be bounded by a function of $k$.

It is quite easy to reduce \kkis\ to \kkhs: for every pair
$(i_1,j_1)$, $(i_2,j_2)$ of adjacent vertices, we need to ensure that
they are not selected simultaneously, which can be forced by a set
that contains every element of rows $i_1$ and $i_2$, except
$(i_1,j_1)$ and $(i_2,j_2)$. However, in Section~\ref{th:closest} we
prove the lower bound for \clstr\ by reduction from a restricted form
of \kkhs\ where each set contains at most one element from each {\em
  row.} The following theorem proves the lower bound for this variant
of \kkhs. The basic idea is that an instance of \kkbis\ can be
transformed in an easy way into an instance of \textsc{Hitting Set}
where each set contains at most one element from each {\em column} and
we want to select exactly one element from each row and each column.
By adding each row as a new set, we can forget about the restriction
that we want to select exactly one element from each row: this
restriction will be automatically satisfied by any solution.
Therefore, we have a \textsc{Hitting Set} instance where we have to
select exactly one element from each column and each set contains at
most one element from each column. By changing the role of rows and
columns, we arrive to a problem of the required form.

\begin{theorem}
\label{th:kkhs}
Assuming the ETH, there is no $2^{o(k\log k)}\cdot n^{O(1)}$ time algorithm for \kkhs,
even  in the special case when each set contains at most one element
from each row.
\end{theorem}
\begin{proof}
To make the notation in the proof less confusing, we
introduce a transposed variant of the problem (denote by \kkhsr), where exactly one
element has to be selected from each column. We prove the lower bound
for \kkhsr\ with the additional restriction that each set contains at
most one element from each column; this obviously implies the
theorem. 

Given an instance $I$ of \kkbis, we construct an equivalent
\tkkhsr\ instance $I'$ on the universe $[2k]\times
[2k]$. For $1\le i \le k$, let set $S_i$ contain the first $k$ elements of
row $i$ and for $k+1 \le i \le 2k$, let set $S_i$ contain the last
$k$ elements of row $i$. For every edge $e$ in instance $I$, we construct a set
$S_e$ the following way. By the way \kkbis\ is defined, we need to
consider only edges connecting some $(i_1,j_1)$ and $(i_2,j_2)$ with
$i_1,j_1\le k$ and $i_2,j_2\ge k+1$. For such an edge $e$, let us define
\[
S_e=\{(i_1,j')\mid 1\le j' \le k, j'\neq j_1\} \cup 
\{(i_2,j')\mid k+1\le j' \le 2k, j'\neq j_2\}.
\]

Suppose that $(1,\delta(1))$, $\dots$, $(2k,\delta(2k))$ is a solution
of $I$ for some permutation $\rho$ of $[2k]$. We claim that it is a
solution of $I'$. As $\rho$ is a permutation, the set satisfies the
requirement that it contains exactly one element from each column. As
$\delta(i)\le k$ if and only if $i\le k$, the set $S_i$ is hit for
every $1\le i\le 2k$. Suppose that there is an edge $e$ connecting
$(i_1,j_1)$ and $(i_2,j_2)$ such that set $S_e$ of $I'$ is not hit by
this solution. Elements $(i_1,\delta(i_1))$ and $(i_2,\delta(i_2))$
are selected and we have $1\le \delta(i_1) \le k$ and $k+1 \le
\delta(i_2) \le 2k$. Thus if these two elements do not hit $S_e$, then
this is only possible if $\delta(i_1)=j_1$ and $\delta(i_2)=j_2$.
However, this means that the solution for $I$ contains the two
adjacent vertices $(i_1,j_1)$ and $(i_2,j_2)$, a contradiction.

Suppose now that $(\rho(1),1)$, $\dots$, $(\rho(2k),2k)$ is a solution
for $I'$. Because of the sets $S_i$, $1\le i \le 2k$, the solution
contains exactly one element from each row, i.e., $\rho$ is a
permutation of $2k$.  Moreover, the sets $S_1$, $\dots$, $S_k$
have to be hit by the $k$ elements in the first $k$ columns. This means
that $\rho(i)\le k$ if $i\le k$ and consequently $\rho(i)>k$ if $i>k$.
We claim that $(\rho(1),1)$, $\dots$, $(\rho(2k),2k)$ is also a solution of
$I$. It is clear that the only thing that has to be verified is that
these $2k$ vertices form an independent set.  Suppose that
$(\rho(j_1),j_1)$ and $(\rho(j_2),j_2)$ are connected by an edge $e$.
We can assume that $\rho(j_1)\le k$ and $\rho(j_2)>k$, which implies
$j_1\le k$ and $j_2>k$. The solution for $I'$ hits set $S_e$, which
means that either the solution selects an element $(\rho(j_1),j')$ or
an element $(\rho(j_2),j')$. Elements $(\rho(j_1),j_1)$ and
$(\rho(j_2),j_2)$ are the only elements of this form in the solution,
but neither of them appears in $S_e$. Thus $(\rho(1),1)$, $\dots$,
$(\rho(2k),2k)$ is indeed a solution of $I$
\end{proof}

\section{Closest String}\label{sec:closest-substring}
Computational biology applications often involve long sequences that
have to be analyzed in a certain way. One core problem is finding a
``consensus'' of a given set of strings: a string that is close to
every string in the input. The \clstr\ problem defined below
formalizes this task. 

\problem{\clstr}{Strings $s_1$, $\dots$, $s_t$ over an alphabet
  $\Sigma$ of length $L$ each, an integer $d$}
{$d$}
{Is there a string $s$ of length $L$ such $d(s,s_i)\le d$ for every
  $1\le i \le t$?}
We denote by $d(s,s_i)$ the {\em Hamming distance} of the strings $s$
and $s_i$, that is, the number of positions where they have different
characters. The solution $s$ will be called the {\em center string.}

\clstr\ and its generalizations (\textsc{Closest Substring},
\textsc{Distinguishing (Sub)string Selection}, \textsc{Consensus
  Patterns})  have been thoroughly explored both from the viewpoint of
approximation algorithms and fixed-parameter tractability
\cite{DBLP:journals/siamcomp/MaS09,DBLP:conf/faw/WangZ09,marx-closest-full,MR1984615,506150,chen-ma-wang-closest,MR2000185,MR2071152,MR1994748,MR2062503}.
In particular, Gramm et al. \cite{MR1984615} showed that \clstr\ is
fixed-parameter tractable parameterized by $d$: they gave an algorithm
with running time $O(d^d\cdot |I|)$. The algorithm works over an
arbitrary alphabet $\Sigma$ (i.e., the size of the alphabet is part of
the input). It is an obvious question whether the dependence on $d$
can be reduced to single exponential, i.e., whether the running time
can be improved to $2^{O(d)}\cdot |I|^{O(1)}$. For small fixed
alphabets, Ma and Sun~\cite{DBLP:journals/siamcomp/MaS09} achieved
single-exponential dependence on $d$: the running time of their
algorithm is $|\Sigma|^{O(d)}\cdot |I|^{O(1)}$. Improved algorithms with
running time of this form, but with better constants in the exponent
were given in \cite{DBLP:conf/faw/WangZ09,chen-ma-wang-closest}.  We
show here that the $d^d$ and $|\Sigma|^d$ dependence are best possible
(assuming the ETH): the dependence cannot be improved to $2^{o(d\log d)}$
or to $2^{o(d\log |\Sigma|)}$. More precisely, what our proof actually shows is that $2^{o(t\log t)}$ dependence is not possible for the parameter $t=\max\{d,|\Sigma|\}$. In particular, single exponential
dependence on $d$ cannot be achieved if the alphabet size is
unbounded.

\begin{theorem}\label{th:closest}
Assuming the ETH, there is no $2^{o(d\log d)}\cdot |I|^{O(1)}$ or
$2^{o(d\log|\Sigma|)}\cdot |I|^{O(1)}$ time algorithm for \clstr.
\end{theorem}
\begin{proof}
We prove the theorem by a reduction from the \textsc{Hitting Set}
problem considered in Theorem~\ref{th:kkhs}. Let $I$ be an instance of
\kkhs\ with sets $S_1$, $\dots$, $S_m$; each set contains at most one
element from each row. We construct an instance $I'$ of \clstr\ as
follows. Let $\Sigma=[2k+1]$, $L=k$, and $d=k-1$ (this means that the
center string has to have at least one character common with every
input string). Instance $I'$ contains $(k+1)m$ input strings
$s_{x,y}$ ($1\le x \le m$, $1\le y \le k+1$). If set $S_x$ contains
element $(i,j)$ from row $i$, then the $i$-th character of
$s_{x,y}$ is $j$; if $S_x$ contains no element of row $i$, then the
$i$-th character of $s_{x,y}$ is $y+k$. Thus string $s_{x,y}$ describes
the elements of set $S_x$, using a certain dummy value between $k+1$ and $2k+1$
to mark the rows disjoint from $S_x$. The strings $s_{x,1}$, $\dots$, $s_{x,k+1}$
differ only in the choice of the dummy values.

We claim that $I'$ has a solution if and only if $I$ has. Suppose that
$(1,\rho(1))$, $\dots$, $(k,\rho(k))$ is a solution of $I$ for some
mapping $\rho:[k]\to [k]$. Then the center string
$s=\rho(1)\dots\rho(k)$ is a solution of $I'$: if element $(i,\rho(i))$
of the solution hits set $S_x$ of $I$, then both $s$ and $s_{x,y}$ have
character $\rho(i)$ at the $i$-th position. For the other direction,
suppose that center string $s$ is a solution of $I'$. As the length of
$s$ is $k$, there is a $k+1\le y \le 2k+1$ that does not appear in
$s$. If the $i$-th character of $s$ is some $1\le c \le k$, then let
$\rho(i)=c$; otherwise, let $\rho(i)=1$ (or any other arbitrary
value). We claim that $(1,\rho(1))$, $\dots$, $(k,\rho(k))$ is a
solution of $I$, i.e., it hits every set $S_x$ of $I$. To see this,
consider the string $s_{x,y}$, which has at least one character common
with $s$. Suppose that character $c$ appears at the $i$-th position in
both $s$ and $s_{x,y}$. It is not possible that $c>k$: character $y$ is the
only character larger than $k$ that appears in $s_{x,y}$, but $y$ does
not appear in $s$. Therefore, we have $1\le c \le k$ and $\rho(i)=c$,
which means that element $(i,\rho(i))=(i,c)$ of the solution hits
$S_x$.

The claim in the previous paragraph shows that solving instance $I'$
using an algorithm for \clstr\ solves the \kkhs\ instance $I$. Note
that the size $n$ of the instance $I'$ is polynomial in $k$ and $m$.
Therefore, a $2^{o(d\log d)}\cdot |I|^{O(1)}$ or a $2^{o(d\log
  |\Sigma|)}\cdot |I|^{O(1)}$ algorithm for \clstr\ would give a
$2^{o(k\log k)}\cdot (km)^{O(1)}$ time algorithm for \kkhs, violating the ETH
(by Theorem~\ref{th:kkhs}).
\end{proof}



\section{Distortion}\label{sec:distortion}

Given an undirected graph $G$ with the vertex set $V(G)$
 and the edge set $E(G)$, a
metric associated with $G$ is $M(G) = (V(G),D)$, where the
distance function $D$ is the shortest path distance between $u$
and $v$ for each pair of vertices $u,v \in V(G)$. We refer to $M(G)$
as to the  {\em graph metric} of $G$. 
 Given a graph metric $M$ 
and another metric space $M'$ with distance functions $D$ and
$D'$, a mapping $f:M \rightarrow M'$ is called an {\em embedding}
of $M$ into $M'$. The mapping $f$ has {\em contraction} $c_f$ and
{\em expansion} $e_f$ if for every pair of points $p,q$ in $M$,
$
D(p,q) \leq D'(f(p),f(q)) \cdot c_f $ and $D(p,q) \cdot e_f \geq
D'(f(p),f(q))$ respectively. We say that $f$ is
\emph{non-contracting} if $c_f$ is at most $1$. A non-contracting
mapping $f$ has \emph{distortion} $d$ if $e_f$ is at most $d$. One of 
the most well studied case of graph embedding is when the host metric $M'$ 
is $\mathbb{R}^1$ and $D'$ is the Euclidean distance. This is also called embedding 
the graph into integers or line. Formally, the problem of \dist\ is defined as follows. 
\problem{\dist}
{A graph $G$, and a positive integer $d$}{$d$}{Is there an embedding $g:V(G) \rightarrow Z$
such that for all $u,v \in V(G)$,  $D(u,v) \leq |g(u)-g(v)| \leq d \cdot D(u,v)$?}

The problem of finding   embedding with good distortion  between metric 
spaces is a fundamental mathematical problem 
\cite{Indyk01,Linial02} that has been studied intensively  \cite{BadoiuCIS05,BadoiuDGRRRS05,BadoiuIS07,KenyonRS04}. 
Embedding a graph metric into a simple low-dimensional metric space
like the real line has proved to be a useful algorithmic tool in
various fields  (for an example see~\cite{GuptaNRS04} for a long list of applications). 
B{\u{a}}doiu \textit{et al.}~\cite{BadoiuDGRRRS05} studied \dist\ from the viewpoint of approximation algorithms 
and exact algorithms. They showed  
that there is a constant $a>1$, such that
$a$-approximation of the minimum distortion of embedding into the line, is
NP-hard and provided an exact algorithm
computing embedding of a $n$ vertex graph into line with distortion $d$ in time $n^{O(d)}$.
Subsequently, Fellows et al.
\cite{DBLP:conf/icalp/FellowsFLLRS09} improved the running time of their algorithm to 
$d^O(d)\cdot n$ and thus proved \dist\ to be fixed parameter tractable parameterized by 
$d$. 
%
We show here
that the $d^{O(d)}$ dependence in the running time of \dist\ algorithm
is optimal (assuming the ETH). To achieve this we first obtain a lower
bound on an intermediate problem called \cperm, then give a reduction
that transfers the lower bound from \cperm\ to \dist. The superexponential dependence on $d$ is 
particularly interesting, as $c^{n}$ time algorithms for finding a minimum distortion 
embedding of a graph on $n$ vertices into line  have been given by 
Fomin et al.~~\cite{FominLS09} and Cygan and Pilipczuk~\cite{abs-1004-5012}.
\problem{\cperm}
{Subsets $S_1$, $\dots$, $S_m$ of $[k]$}{$k$}
{A permutation $\rho$ of $[k]$ such that for every $1\le i \le m$,
  there is a $1\le j <k$ such that $\rho(j),\rho(j+1)\in S_i$.}

Given a permutation $\rho$ of $[k]$, we say that $x$ and $y$ are
{\em neighbors} if $\{x,y\}=\{\rho(i),\rho(i+1)\}$ for some $1\le i<
k$. In the \cperm\ problem the task is to find a permutation that hits
every set $S_i$ in the sense that there is a pair $x,y\in S_i$ that
are neighbors in $\rho$. 
\begin{theorem}
\label{th:cperm}
Assuming the ETH, there is no $2^{o(k\log k)}m^{O(1)}$ time algorithm for
\cperm.
\end{theorem}
\begin{proof}
  Given an instance $I$ of \kkbis, we construct an equivalent instance
  $I'$ of \cperm. Let $k'=24k$ and for ease of notation, let us
  identify the numbers in $[k']$ with the elements $r^\ell_i$, $\bar
  r^\ell_i$, $c^\ell_j$, $\bar c^\ell_j$ for $1\le \ell \le 3$, $1\le
  i,j\le 2k$. The values $r^\ell_i$ represent the rows and the values
  $c^\ell_j$ represent the columns. If $\bar r^\ell_i$ and $c^\ell_j$
  are neighbors in $\rho$, then we interpret it as selecting element
  $j$ from row $i$. More precisely, we want to construct the sets
  $S_1$, $\dots$, $S_m$ in such a way that if $(1,\delta(1))$,
  $\dots$, $(2k,\delta(2k))$ is a solution of $I$, then the following
  permutation $\rho$ of $[k']$ is a solution of $I'$:
\begin{gather*}
r^1_1,\bar r^1_1, c^1_{\delta(1)}, \bar c^1_{\delta(1)},
r^1_2,\bar r^1_2, c^1_{\delta(2)}, \bar c^1_{\delta(2)},
\dots,
r^1_{2k},\bar r^1_{2k}, c^1_{\delta(2k)}, \bar c^1_{\delta(2k)},\\
r^2_1,\bar r^2_1, c^2_{\delta(1)}, \bar c^2_{\delta(1)},
r^2_2,\bar r^2_2, c^2_{\delta(2)}, \bar c^2_{\delta(2)},\dots,
r^2_{2k},\bar r^2_{2k}, c^2_{\delta(2k)}, \bar c^2_{\delta(2k)},\\
r^3_1,\bar r^3_1, c^3_{\delta(1)}, \bar c^3_{\delta(1)},
r^3_2,\bar r^3_2, c^3_{\delta(1)}, \bar c^3_{\delta(2)},\dots,
r^3_{2k},\bar r^3_{2k}, c^3_{\delta(2k)}, \bar c^3_{\delta(2k)}.
\end{gather*}
The first property
that we want to ensure is that every solution of $I'$ looks roughly
like $\rho$ above: pairs $r^\ell_i\bar r^\ell_i$ and pairs
$c^\ell_j\bar c^\ell_j$ alternate in some order. Then we can define a
permutation $\delta$ such that $\delta(i)=j$ if $r^1_i\bar
r^1_i$ is followed by the pair $c^1_j\bar c^1_j$. The sets in instance
$I'$ will ensure that this permutation $\delta$ is a solution of
$I$. Let instance $I'$ contain the following groups of sets:
\begin{enumerate}
\item For every $1\le \ell \le 3$ and
  $1\le i \le 2k$, there is a set $\{r^\ell_i,\bar r^\ell_i\}$ ,
\item For every $1\le \ell \le 3$ and
  $1\le j \le 2k$, there is a set $\{c^\ell_j,\bar c^\ell_j\}$, 
\item For every $1\le \ell'< \ell''\le 3$, $1\le i \le 2k$,
  $X\subseteq [2k]$, there is a set $\{\bar r^{\ell'}_i,\bar r^{\ell''}_i\} \cup \{c^{\ell'}_j \mid j\in X\}\cup \{c^{\ell''}_j
  \mid j\not\in X\}$,
\item For every $1\le
  i \le k$, there is a set $\{\bar r^1_i\}\cup\{ c^1_j\mid 1\le j \le k\}$,
\item For every $k+1\le
  i \le 2k$, there is a set $\{\bar r^1_i\}\cup\{ c^1_j\mid k+1\le j \le 2k\}$,
\item For every two adjacent vertices $(i_1,j_1)\in I_1$ and
  $(i_2,j_2)\in I_2$, there is a set $\{\bar r^1_{i_1},\bar
  r^1_{i_2}\}\cup \{ c^1_j\mid 1\le j \le k, j\neq j_1\}\cup\{
  c^1_j\mid k+1\le j \le 2k, j\neq j_2\}$.

\end{enumerate}
Recall that every edge of instance $I$ goes between the independent
sets $I_1=\{(i,j)\mid i,j\le k\}$ and $I_2=\{(i,j)\mid i,j\ge k+1\}$.
Let us verify first that if $\delta$ is a solution of $I$, then the
permutation $\rho$ described above satisfies every set. It is clear
that sets in the first two groups are satisfied. To see that every set
in group 3 is satisfied, consider a set corresponding to a particular
$1 \le \ell' < \ell'' \le 3$, $1\le i \le 2k$, $X\subseteq [2k]$. If
$\delta(i)\in X$, then $\bar r^{\ell'}_i$ and $c^{\ell'}_{\delta(i)}$
are neighbors and both appear in the set; if $\delta(i)\not \in X$,
then $\bar r^{\ell''}_i$ and $c^{\ell''}_{\delta(i)}$ are neighbors
and both appear in the set. Sets in group 4 and 5 are satisfied
because $\delta(i)\le k$ for $1\le i \le k$ and $\delta(i)\ge k+1$ for
$k+1\le i \le 2k$. Finally, let $(i_1,j_1)\in V_1$ and $(i_2,j_2)\in
V_2$ be two adjacent vertices and consider the corresponding set in
group 6.  As the solution of $I$ is an independent set, either
$\delta(i_1)\neq j_1$ or $\delta(i_2)\neq j_2$. In the first case,
$\bar r^1_{i_1}$ and $c^1_{\delta(i_1)}$ are neighbors and both appear
in the set; in the second case, $\bar r^1_{i_2}$ and
$c^1_{\delta(i_2)}$ are neighbors and both appear in the set.

Next we show that if $\rho$ is a solution of $I'$, then a solution for
$I$ exists. We say that an element $\bar r^\ell_i$ is {\em good} if
its neighbors are $r^\ell_i$ and $c^{\ell'}_j$ for some $1\le \ell'
\le 3$ and $1\le j \le 2k$. Similarly, an element $c^{\ell}_j$ is good
if its neighbors are $\bar c^\ell_j$ and $\bar r^{\ell'}_i$ for some
$1\le \ell' \le 3$ and $1\le i \le 2k$. Our first goal is to show that
every $\bar r^\ell_i$ and $c^\ell_j$ is good. The sets in group 1 and
2 ensure that $r^\ell_i$ and $\bar r^\ell_i$ are neighbors,  and
$c^\ell_j$ and $\bar c^\ell_j$ are neighbors.

We claim that for every $1\le \ell'<\ell'' \le 3$, and $1\le i \le
2k$, if elements $\bar r^{\ell'}_i$ and $\bar r^{\ell''}_i$ are not neighbors,
then both of them are good.  Let us build a $4k$-vertex graph $B$ whose vertices
are $c^{\ell'}_j$, $c^{\ell''}_j$ ($1\le j \le 2k$).  Let us connect
by an edge those vertices that are neighbors in $\rho$.  Moreover, let
us make $c^{\ell'}_j$ and $c^{\ell''}_j$ adjacent for every $1 \le j
\le 2k$. Observe that the degree of every vertex is at most 2 (as
$c^{\ell'}_j$ has only one neighbor besides $\bar c^{\ell'}_j$).
Moreover, $B$ is bipartite: in every cycle, edges of the form
$c^{\ell'}_jc^{\ell''}_j$ alternate with edges not of this form.
Therefore, there is a bipartition $(Y,\bar Y)$ of $B$ such that the set
$Y$ (and hence $\bar Y$)  contains
exactly one of $c^{\ell'}_j$ and $c^{\ell''}_j$ for every $1\le j \le
2k$.  Group 3
contains a set $S_Y=\{\bar r^{\ell'}_i,\bar r^{\ell''}_i\} \cup Y$ and
a set $S_{\bar Y}=\{\bar r^{\ell'}_i,\bar r^{\ell''}_i\} \cup \bar Y$:
as $Y$ contains exactly one of $c^{\ell'}_j$ and $c^{\ell''}_j$, there is a
choice of $X$ that yields these sets. Permutation $\rho$ satisfies
$S_Y$ and $S_{\bar Y}$, thus each of $S_Y$ and $S_{\bar Y}$ contains a
pair of neighboring elements. By assumption, this pair cannot be
$\bar r^{\ell'}_i$ and $\bar r^{\ell''}_i$. As $Y$ induces an independent
set of $B$, this pair cannot be contained in $Y$ either. Thus the only
possibility is that one of $\bar r^{\ell'}_j$ and $\bar r^{\ell''}_j$ is the
neighbor of an element of $Y$. If, say, $\bar r^{\ell'}_j$ is a neighbor of
an element $y\in Y$, then $\bar r^{\ell'}_j$ is good. In this case,
$\bar r^{\ell'}_j$ is not the neighbor of any element of $\bar Y$, which
means that the only way two members of $S_{\bar Y}$ are neighbors if
$\bar r^{\ell''}_j$ is a neighbor of a member of $\bar Y$, i.e.,
$\bar r^{\ell''}_j$ is also good.

At most one of $\bar r^{2}_i$ and $\bar r^{3}_i$ can be the neighbor
of $\bar r^{1}_i$, thus we can assume that $\bar r^{1}_i$ and $\bar
r^{\ell}_i$ are not neighbors for some $\ell\in \{2,3\}$. By the claim
in the previous paragraph, $\bar r^{1}_i$ and $\bar r^{\ell}_i$ are
both good. In particular, this means that $\bar r^{1}_i$ is not the
neighbor of $\bar r^{2}_i$ and $\bar r^3_i$, hence applying again the
claim, it follows that $\bar r^{2}_i$ and $\bar r^3_i$ are both good.
Thus $\bar r^{\ell}_i$ is good for every $1\le \ell \le 3$ and $1\le i
\le 2k$, and the pigeonhole principle implies that $c^{\ell}_j$ is
good for every $1\le \ell \le 3$ and $1\le i \le 2k$.

As every $c^{1}_j$ is good, the sets in groups 4 and 5 can be
satisfied only if every $\bar r^{1}_i$ has a neighbor $c^{1}_j$. Let
$\delta(i)=j$ if $c^1_j$ is the neighbor of $\bar r^{1}_i$; clearly
$\delta$ is a permutation of $[2k]$. We claim that $\delta$ is a
solution of $I$. The sets in group 4 and 5 ensure that $\delta(i)\le
k$ for every $1\le i \le k$ and $\delta(i)\ge k+1$ if $k+1\le i \le
2k$. To see that $(1,\delta(i))$, $\dots$, $(2k,\delta(2k))$ is an
independent set, consider two adjacent vertices $(i_1,j_1)\in I_1$ and
$(i_2,j_2)\in I_2$. We show that it is not possible that
$\delta(i_1)=j_1$ and $\delta(i_2)=j_2$. Consider the set $S$ in group
6 corresponding to the edge connecting $(i_1,j_1)$ and $(i_2,j_2)$. As
$\bar r^1_{i_1}$, $\bar r^1_{i_2}$, and every $c^1_j$ is good, then
only way $S$ is can be satisfied is that $\bar r^1_{i_1}$ or $\bar
r^1_{i_2}$ is the neighbor of some $c^1_j$ appearing in $S$. If
$\delta(i_1)=j_1$ and $\delta(i_2)=j_2$, then the $c^1_{j_1}$ and
$c^1_{j_2}$ are the neighbors of $\bar r^1_{i_1}$ and $\bar
r^1_{i_2}$, respectively, but $c^1_{j_1}$ and $c^1_{j_2}$ do not
appear in $S$. This shows that if there is a solution for $I'$, then there
is a solution for $I$ as well.

The size of the constructed instance $I'$ is polynomial in $2^k$. Thus
if $I'$ can be solved in time $2^{o(k'\log k')}\cdot |I'|
=2^{o(k\log k)}\cdot 2^{O(k)}=2^{o(k\log k)}$, then
this gives a $2^{o(k\log k)}$ time algorithm for \kkbis.
\end{proof}

\begin{theorem}
\label{th:dist}
Assuming the ETH, there is no $2^{o(d\log d)}\cdot n^{O(1)}$ time algorithm for \dist.
\end{theorem}
\begin{figure}[t]
 \begin{center}
\includegraphics[scale=0.8]{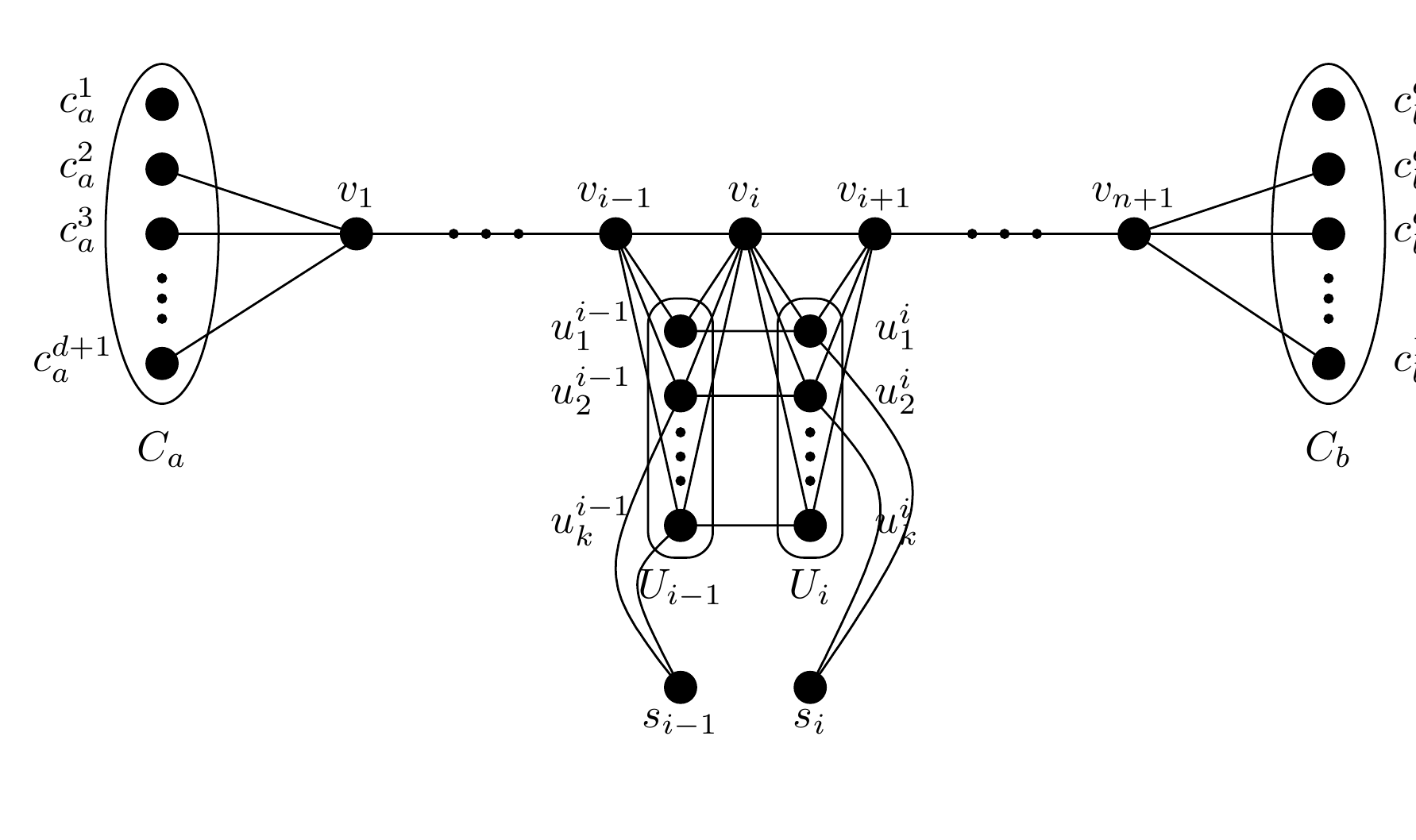}
\end{center}
\label{fig:distortion}
\caption{A construction used in Theorem~\ref{th:dist}}
\end{figure}

\begin{proof}
We prove the theorem by a reduction from the \cperm\ problem. Let $I$ be an instance of \cperm\, consisting 
of subsets $S_1$, $\dots$, $S_m$ of $[k]$. Now we show how to construct the graph $G$, an input to  
\dist\ corresponding to $I$.  
For an ease of presentation we identify $[k]$ with vertices 
$u_1,\ldots, u_k$. We also set $U=\{u_1,\ldots, u_k\}$ and $d=2k$. The vertex set of $G$ consists of the following set of vertices.  
\begin{itemize}
 \setlength\itemsep{-.9mm}
\item A vertex $u^{i}_j$ for every $1\leq i \leq m$ and  $1\leq j\leq k$. We also denote the set $\{u_1^i,\ldots ,u_k^i\}$ by $U_i$. 
\item A vertex $s_i$ for each set $S_i$. 
\item Two cliques $C_a$ and $C_b$ of size $d+1$ consisting of vertices $c_a^1,\ldots,c_a^{d+1}$ and $c_b^1,\ldots,c_b^{d+1}$ respectively. 
\item A path $P$ of length $m$ (number of edges) consisting of vertices $v_1,\ldots,v_{m+1}$.  
\end{itemize}
We add the following more edges among these vertices. We add edges from all the vertices in clique $C_a$ but $c_a^1$ to $v_1$ 
and add edges from all the vertices in clique $C_b$ but $c_b^1$ to $v_{m+1}$. 
For all 
$1\leq i < m$ and  $1\leq j\leq k$, make $u^{i}_j$ adjacent to $v_i$, $v_{i+1}$ and $u^{i+1}_j$. For $1\leq j\leq k$, make $u^m_j$ adjacent to  
$v_m$, $v_{m+1}$. Finally make $s_i$ adjacent to $u_j^i$ if $u_j\in S_i$. 
This concludes the construction. A figure corresponding to the construction can be found in Figure~\ref{fig:distortion}.

For our proof of correctness we also need the following known facts about distortion $d$ embedding of a graph into integers. 
For an embedding $g$, let $v_1$, $v_2,\dots,v_q$ be an ordering of
the vertices such that $g(v_1) < g(v_2) < \dots < g(v_n)$. If $g$ is such that for all $1\leq i < q$, 
$D(v_i,v_{i+1})=|g(v_i)-g(v_{i+1})|$, then the mapping $g$ is called {\em pushing embedding}. It is known 
that pushing embeddings are always non-contracting and if $G$ can be embedded into integers with 
distortion $d$, then there is a pushing embedding of $G$ into integers with distortion $d$~\cite{DBLP:conf/icalp/FellowsFLLRS09}.

Let a permutation $\rho$ of $[k]=U$ be a solution to $I$, an instance of \cperm. This automatically leads to a permutation 
on $U$ that we represent by $\rho(U)$. There is a natural bijection between $U$ and $U_i$ 
with $u_j\in U$ being mapped to $u_j^i$. So when we write $\rho(U_i)$ then this means that 
the vertices of $U$ are permuted with respect to $\rho$ and being identified with its counterpart in $U_i$. 
Now we give a pushing embedding for the vertices in $G$ with $c_a^1$ being 
placed at $0$.   
All the vertices except the set vertices $s_i$ appear in the following order 
$$c_a^1,\ldots,c_a^{d+1},v_1,\rho(U_1),v_2,\rho(U_2),v_3,\ldots,v_m,\rho(U_m),v_{m+1},c_b^{d+1},\ldots,c_b^1.$$
Since $\rho $ is a solution to $I$ we know that for every $S_i$ there exists a $1\leq j <k$ such that $\rho(j)\rho(j+1)\in S_i$. 
We place $s_i$ between $\rho(u_j^i)$ and $\rho(u_{j+1}^i)$. By our construction the given embedding is pushing and hence non-contracting. 
To show that for every pair of vertices $u,v \in V(G)$, $|g(u)-g(v)|\le d\cdot D(u,v)$, we only have to show that for every edge $uv\in E(G)$, 
$|g(u)-g(v)|\leq d$. This can be readily checked from the construction. What needs to be verified is that for any two adjacent vertices $u$ and $v$, the sequence of vertices between $u$ and $v$ in the pushing embedding give a total distance at most $d\cdot D(u,v)$. 
The cruical observation is that the distance between two consecutive vertices from $U_i$ is $2$, and hence it must be at least distance 
$2$ apart on the line. If $s_i$ is adjacent to two consecutive vertices in $U_i$ we can ``squeeze'' in $s_i$ between those two vertices without 
disturbing the rest of the construction.

In the reverse direction, assume that we start with a distortion $d$ pushing embedding of $G$.   
Consider the layout of the graph induced on $C_a$ and the vertex $v_1$. This is a clique of size $d+2$ 
minus an edge and hence $C_a \cup \{v_1\}$ can be layed out in two ways:  
$c_a^1,C_a\setminus \{c_a^1\}, v_1$ or  $v_1,C_a\setminus \{c_a^1\}, c_a^1$. Since we can 
reverse the layout, we can assume without loss of generality that it is $c_a^1,C_a\setminus \{c_a^1\}, v_1$. Without loss of generality 
we can also assume that $v_1$ is placed on position $0$. Since every vertex in $U_1$ is adjacent to $v_1$ and the negative positions are 
taken by the vertices in $C_a$, the $k=d/2$ vertices 
of $U_1$ must lie on the positions $\{1,\ldots,d\}$. We first argue that no vertex of $U_1$ occupies the position $d$. Suppose it does. Then the 
rightmost vertex of $U_2$ (to the right of $v_1$ in the embedding) must be on position at least $2d$. Simultaneously $v_2$ must be 
on position at most $d-1$ since $d$ is already occupied and $v_2$ is adjacent to $v_1$. But $v_2$ is adjacent to the rightmost vertex of $U_2$ and hence the distance on the line between them becomes at least $d+1$, a contradiction. So $U_1$ must use only positions in $\{1,\ldots,d-1\}$. Since the distance between two consecutive vertices in $U_1$ is $2$ together with the fact that we started with a pushing embedding imply that 
the vertices of $U_1$ occupy all odd positions of $\{1,\ldots ..,d-1\}$. Now, $U_2$ must be on the positions 
in $\{d+1,\ldots ,2d\}$ with the rightmost vertex in $U_2$ being on at least $2d-1$. Since $d-1$ is occupied by someone in $U_1$ and $v_2$ 
is adjacent to both $v_1$ and the rightmost vertex of $U_2$ it follows that $v_2$ must be on position $d$. 

We can now argue similarly to the previous paragraph that $U_2$ does not use position $2d$, and hence $v_3$ is on position $2d$ while $U_2$ 
must use the odd positions of $\{d+1,\ldots,2d-1\}$. We can repeat this argument for all $i$ and position the vertex  $v_i$ of the path at 
$d(i-1)$ and place the vertices of $U_i$ at odd positions between $d(i-1)$ and $di$. Of course, all the vertices of the clique 
$C_b$ will come after $v_{m+1}$. 

Consider the order in which the embedding puts the vertices of $U_1$. We claim that it 
must put the vertices of $U_2$ in the same order. Look at the embedding of $U_1$ and $U_2$ from left to right and let $j$ be the first index 
where $u_\alpha^1$ of $U_1$ is placed between $0$ and $d$ while $u_\beta^2$ of $U_2$ is placed between $d$ and $2d$ and $\alpha \neq \beta $. This implies that 
$u_\alpha^2$
appears further back in the permutation of $U_2$ and hence the distance between the positions of $u_\alpha^1$ and $u_\alpha^2$ in $U_1$ 
is more than $d$ while $u_\alpha^1$ and $u_\alpha^2$ are adjacent to each other in the graph. By repeating this argument for all $i$ and 
$i+1$ we can show that order of all $U_i$'s is the same. Consider $s_i$. It must be put on some even position, with some vertices of 
$U_j$ coming before and after $s_i$. But then, because we started with pushing embedding we have that 
$s_i$ is adjacent to both those vertices, and hence $i=j$ as $s_i$ is adjacent to only the vertices in $U_i$. 

Now we take the permutation $\rho$ for $[k]$, imposed by the ordering of $U_1$, as a solution to the instance $I$ of \cperm.  For every set 
$S_i$ we need to show that there exists a $1\le j <k$ such that $\rho(j),\rho(j+1)\in S_i$. Consider the corresponding $s_i$ in the embedding 
and look at the vertices that are placed left and right of it. Let these be $u_\alpha^i$ and $u_\beta^i$. Then by construction $\alpha$ and $\beta$ are neighbors  
to $s_i$ in $G$ and hence $\alpha$ and $\beta$ belong to $S_i$. Now since the ordering of $U_i$'s are same we have that they are 
consecutive  in the permutation $\rho$. This concludes the proof in the reverse direction. 

The claim in the previous paragraph shows that an algorithm finding a distortion $d$ embedding of $G$ into  line 
solves the instance $I$ of \cperm. Note the number of vertices in $G$ is bounded by a polynomial in $k$ and $m$. Therefore a $2^{o(d \log d)} \cdot |V(G)|^{O(1)} $ algorithm for \dist\ would give a   
 $2^{o(k \log k)} \cdot (km)^{O(1)} $ algorithm for \cperm, violating the ETH by Theorem~\ref{th:cperm}. 
\end{proof}

\section{Disjoint Paths}
There are many natural graph problems that are fixed-parameter
tractable parameterized by the treewidth of the input graph. In most
cases, these results can be obtained by well-understood dynamic
programming techniques. In fact, Courcelle's Theorem provide a clean
way of obtaining such results. If the dynamic programming needs to
keep track of a permutation, partition, or a matching at each node,
then running time of such an algorithm is typically of the form $w^{O(w)}\cdot n^{O(1)}$
on graphs with treewidth $w$ \cite{Scheffler94}. We demonstrate a problem where this form
of running time is necessary for the solution and it cannot be
improved to $2^{o(w\log w)}\cdot n^{O(1)}$. We start with definitions 
of treewidth and pathwidth. 

\paragraph{Definitions of Treewidth and Pathwidth.}
\label{apx:twpw}
A \emph{tree decomposition} of a graph $G$ is a pair $(\mathcal{X},T)$ where $T$
is a tree and ${\cal X}=\{X_{i} \mid i\in V(T)\}$ is a collection of subsets
of $V$ such that:
\begin{enumerate}
\item  $\bigcup_{i \in V(T)} X_{i} = V$, 
\item  for
each edge $xy \in E$, $\{x,y\}\subseteq X_i$ for some
$i\in V(T)$; 
\item for each $x\in V$ the set $\{ i \mid x \in X_{i} \}$
induces a connected subtree of $T$. 
\end{enumerate}

The \emph{width} of the  tree decomposition is $\max_{i \in V(T)} \{|X_{i}| - 1\}$. The \emph{treewidth} of a
graph $G$ is the minimum width over all tree decompositions of $G$.
We denote by $\tw(G)$ the treewidth  of graph $G$. If in the definition of treewidth we restrict 
the tree $T$ to be a path then we get the notion of pathwidth and denote it by  $\pw(G)$. 

Now we return to our problem. 
Given an undirected graph $G$ and $p$  vertex pairs $(s_i,t_i)$, the \disjpath\ problem asks whether there exists $p$ mutually vertex disjoint paths in $G$ linking these pairs. This is one of the classic problems in combinatorial optimization and algorithmic graph theory,  
and has many applications, for example in transportation networks, VLSI layout, and 
virtual circuits routing in high-speed networks. 
The problem is NP-complete if $p$ is part of the input and remains so 
even if restrict the input graph to be planar~\cite{Karp75,Lynch75}.
However if $p$ is fixed then the problem is famously fixed-parameter
tractable as a consequence of the seminal Graph Minors theory of
Robertson and Seymour~\cite{RobertsonS95b}.   A basic building block in their algorithm for \disjpath\  
is an algorithm for \disjpath\ on graphs of bounded treewidth.
To our interest is the following 
parameterization of  \disjpath.

\problem{\disjpath}
 {A graph $G$ together with a tree-decomposition of width $w$, and $p$  vertex pairs $(s_i,t_i)$.}{$w$}
{Does there exist $p$ mutually vertex disjoint paths in $G$ linking $s_i$ to $t_i$?}

The best known algorithm for this problem runs in time $2^{O(w\log w)}\cdot n$~\cite{Scheffler94} and here 
we show that this is indeed optimal.  To get this lower bound we first give a linear parameter reduction from
 \kkhs\ to \dirpath,   a variant of \disjpath\ where the input is a directed graph, parameterized by pathwidth  
 of the underlying undirected graph. Finally we obtain a lower bound of $2^{o(k\log k)} |V(G)|^{O(1)}$ on \disjpath\  
 parameterized by pathwidth under the ETH,  by giving a linear parameter reduction from \dirpath\  parameterized by 
 pathwidth to \disjpath\  parameterized by pathwidth. Obviously, this proves the 
same lower bound under the (potentially much smaller) parameter treewidth as well.

\begin{theorem}
Assuming the ETH, there is no $2^{o(w\log w)}\cdot n^{O(1)}$ time
algorithm  for \dirpath.
\end{theorem}
\begin{proof}
  The key tool in the reduction from \kkhs\ to \dirpath\ is the
  following gadget. For every $k\ge 1$ and set $S\subseteq [k]\times
  [k]$, we construct the gadget $G_{k,S}$ the following way (see Figure~\ref{fig:path} for illustration).
\begin{figure}
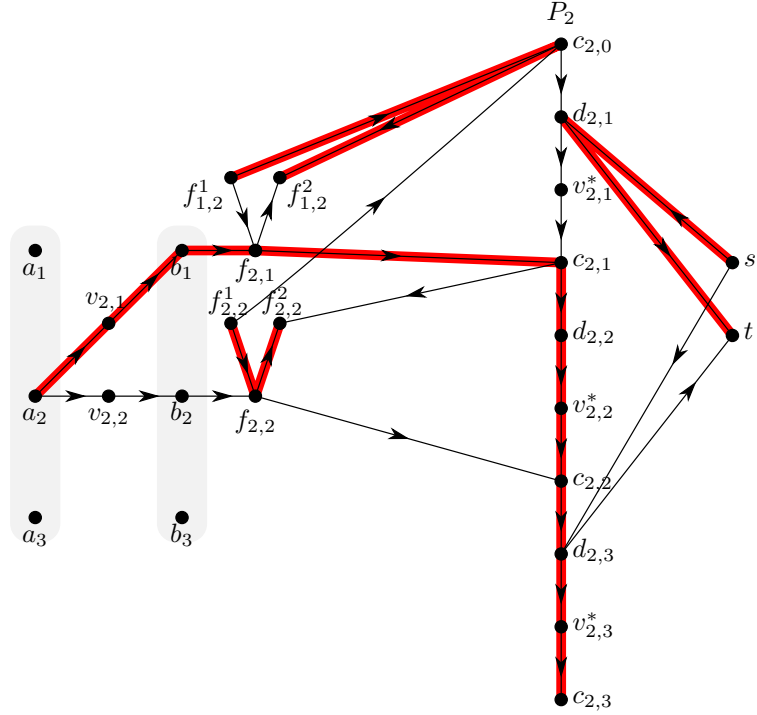

\begin{center}
{\small \svg{0.65\linewidth}{path}}
\caption{Part of the gadget $G_{3,S}$ with $(2,1),(2,3)\in S$. The highlighted paths satisfy the demands $(a_2,c_{2,3})$, $(f^1_{1,2},f^2_{1,2})$, $(f^1_{2,2},f^2_{2,2})$ and $(s,t)$.}\label{fig:path}
\end{center}
\end{figure}
\begin{itemize}
 \setlength\itemsep{-4pt}
\item For every $1\le i \le k$, it contains vertices $a_i$, $b_i$.
\item For every $1\le i, j \le k$, it contains a vertex $v_{i,j}$ and
  edges $\ora{a_iv_{i,j}}$, $\ora{v_{i,j}b_j}$.
\item For every $1\le i \le k$, it contains a directed path
  $P_i=c_{i,0}d_{i,1}v^*_{i,1}c_{i,1}\dots d_{i,k}v^*_{i,k}c_{i,k}$.
\item For every $1\le i, j \le k$, it contains vertices $f_{i,j}$,
  $f^1_{i,j}$, $f^2_{i,j}$ and edges $\ora{b_{j}f_{i,j}}$,
  $\ora{f_{i,j}c_{i,j}}$, $\ora{f^1_{i,j}f_{i,j}}$,
  $\ora{f_{i,j}f^2_{i,j}}$, $\ora{f^1_{i,j}c_{i,0}}$,
  $\ora{c_{i,j-1}f^2_{i,j}}$.
\item It contains two vertices $s$ and $t$, and for every $(i,j)\in S$,
  there are two edges $\ora{sd_{i,j}}$, $\ora{d_{i,j}t}$.
\end{itemize}
The demand pairs in the gadget are as follows:
\begin{itemize}
 \setlength\itemsep{-.9mm}
\item For every $1\le i \le k$, there is a demand $(a_i,c_{i,k})$.
\item For every $1\le i, j \le k$, there is a demand
  $(f^1_{i,j},f^2_{i,j})$.
\item There is a demand $(s,t)$.
\end{itemize}
This completes the description of the gadget. The intuition behind the construction is the following. To satisfy the demand $(a_i,c_{i,k})$, the path needs to leave $a_i$ to $v_{i,j}$ for some $1\le j \le k$.  Thus if a 
collection of paths form a solution for the gadget, then for every
$1\le i \le k$, exactly one of the vertices $v_{i,1}$, $\dots$,
$v_{i,k}$ is used by the paths. We say that a solution {\em represents}
the mapping $\rho: [k]\to [k]$ if for every $1\le i \le k$, vertex $v_{i,\rho(i)}$ is used by
the paths in the solution. Moreover, if the path satisfying $(a_i,c_{i,k})$ leaves $a_i$ to $v_{i,j}$, then it enters the path $P_i$ via the vertex $f_{i,j}$, and reaches $c_{i,k}$ on the path $P_i$. In this case, the demand   $(f^1_{i,j},f^2_{i,j})$ cannot use vertex $f_{i,j}$, and has to use the part of $P_i$ from $c_{i,0}$ to $c_{i,j-1}$. Then these two paths leave free only vertex $v^*_{i,j}$ of $P_i$ and no other $v^*$. This means that the $v_{i,j}$ and $v^*_{i,j}$ vertices behave exactly the opposite way: if $v_{i,\rho(i)}$ is used by the solution, then every vertex $v^*_{i,1}$, $\dots$, $v^*_{i,k}$ is used, with the exception of $v^*_{i,j}$. 
The following claim formalizes this important property of the gadget.
\begin{claim}
\label{claim:dirdispath}
For every $k\ge 1$ and $S\subseteq [k]\times [k]$, gadget $G_{k,S}$
has the following properties:
\begin{enumerate}
 \setlength\itemsep{-4pt}
\item For every $\rho: [k] \to [k]$ that hits $S$, gadget $G_{k,S}$ has a solution that
  represents $\rho$, and $v^*_{i,\rho(i)}$ is not used by the paths in
  the solution for any $1\le i\le k$.
\item If $G_{k,S}$ has a solution that represents $\rho$, then $\rho$
  hits $S$ and vertex $v^*_{i,j}$ is used by the paths in the solution
  for every $1\le i \le k$ and $j\neq \rho(i)$.
\end{enumerate}
\end{claim}
\begin{proof}
  To prove the first statement, we construct a solution the following way.
  Demand $(a_{i},c_{i,j})$ is satisfied by the path
  $a_iv_{i,\rho(i)}b_{\rho(i)}f_{i,\rho(i)}c_{i,\rho(i)}\dots
  c_{i,k}$, where we use a subpath of $P_i$ to go from $c_{i,\rho(i)}$
  to $c_{i,k}$. For every $1\le i,j\le k$, if $j\neq \rho(i)$, then
  the demand $(f^1_{i,j},f^2_{i,j})$ is satisfied by the path
  $f^1_{i,j}f_{i,j}f^2_{i,j}$. If $j=\rho(i)$, then vertex $f_{i,j}$
  is already used by the demand $(a_{i},c_{i,j})$. In this
  case demand $(f^1_{i,j},f^2_{i,j})$ is satisfied by the path
  $f^1_{i,j}c_{i,0}\dots c_{i,j-1}f^2_{i,j}$. Finally, as $\rho$ hits
  $S$, there is a $1\le i \le k$ such that $(i,\rho(i))\in S$ and
  hence the edges $\ora{sd_{i,\rho(i)}}$ and $\ora{d_{i,\rho(i)}t}$
  exist. Therefore, we can satisfy the demand $(s,t)$ via
  $d_{i,\rho(i)}$. Note that this vertex is not used by the other
  paths: the path satisfying demand $(a_{i},c_{i,k})$ uses $P_i$ only from
  $c_{i,\rho(i)}$ to $c_{i,k}$, the path satisfying demand
  $(f^1_{i,\rho(i)},f^2_{i,\rho(i)})$ uses $P_i$ from $c_{i,0}$ to
  $c_{i,\rho(i)-1}$, and no other path reaches $P_i$. This also
  implies that $v^*_{i,\rho(i)}$ is used by none of the paths, as
  required.

  For the second part, consider a solution of $G_{k,S}$ representing
  some mapping $\rho$. This means that the path of demand
  $(a_i,c_{i,k})$ uses vertex $v_{i,\rho(i)}$ and hence $b_{i,\rho(i)}$. The
  only way to reach $c_{i,k}$ from $b_{i,\rho(i)}$ without going through any
  other terminal vertex is using the path $f_{i,\rho(i)}c_{i,\rho(i)}\dots
  c_{i,k}$. This means that demand $(f^1_{i,\rho(i)},f^2_{i,\rho(i)})$
  cannot use vertex $f_{i,\rho(i)}$, hence it has to use the path
  $f^1_{i,\rho(i)}c_{i,0}\dots c_{i,\rho(i)-1} f^2_{i,\rho(i)}$. It
  follows that for every $1\le i \le k$ and $j\neq \rho(i)$, vertices
  $d_{i,j}$ and $v^*_{i,j}$ are used by the paths satisfying demands
  $(a_i,c_{i,k})$ and $(f^1_{i,\rho(i)},f^2_{i,\rho(i)})$. This shows that every
  $v^*_{i,j}$ with $j\neq \rho(i)$ is used by the paths in the
  solution.  Moreover, the path satisfying $(s,t)$ has to go through
  vertex $d_{i,\rho(i)}$ for some $i$. By the way the edges incident to $s$ and $t$
  are defined, this is only possible if $\rho(i)\in S$, that is,
  $\rho$ hits $S$.
\end{proof}

Let $S_1$, $\dots$, $S_m$ be the sets appearing in the \kkhs\
instance $I$. We construct an instance $\vec I$ of \dirpath\ consisting of $m$
gadgets $G_1$, $\dots$, $G_m$, where gadget $G_t$ $(1\le t \le m)$ is a copy of the
gadget $G_{k,S_i}$ defined above. For every $1\le t < m$ and every
$1\le i,j \le k$,  we identify vertex $v^*_{i,j}$ of $G_t$ and vertex
$v_{i,j}$ of $G_{t+1}$. This completes the description of the
instance $\vec I$ of \dirpath.

We have to show that the pathwidth of the constructed graph $\vec G$ of $\vec I$ is $O(k)$ and
that $\vec I$ has a solution if and only if $I$ has. To bound the
pathwidth of $\vec G$, for every $0\le t \le m$, $1\le i,j\le k$, let us define the bag $B_{t,i,j}$ such that it
contains vertices $a_1$, $\dots$, $a_k$, $b_1$, $\dots$, $b_k$, $s$,
$t$, $f_{i,j}$, $f^1_{i,j}$, $f^2_{i,j}$, and the path
$P_i$ of gadget $G_t$ (unless $t=0$), and vertices $a_1$, $\dots$,
$a_k$, $b_1$, $\dots$, $b_k$ of gadget $G_{t+1}$ (unless $t=m$). It can be easily
verified that the size of each bag is $O(k)$ and if two vertices are
adjacent, then they appear together in some bag. Furthermore, if we
order the bags lexicographically according to $(t,i,j)$, then each
vertex appears precisely in an interval of the bags. This shows that
the pathwidth of $\vec G$ is $O(k)$.

Next we show that if $I$ has a solution $\rho: [k]\to [k]$, then $\vec I$ also has a
solution. As $\rho$ hits every $S_t$, by the first part of the Claim, each gadget $G_t$ has a solution
representing $\rho$. To combine these solutions into a solution for
$\vec I$, we have to make sure that the vertices $v_{i,j}$, $v^*_{i,j}$
that were identified are used only in one gadget. Since the solution
for gadget $G_t$ represents $\rho$, it uses vertices $v_{1,\rho(i)}$,
$\dots$, $v_{k,\rho(k)}$, but no other $v_{i,j}$ vertex. As vertex
$v_{i,j}$ of gadget $G_t$ was identified with vertex $v^*_{i,j}$ of
gadget $G_{t-1}$, these vertices might be used by the solution of
$G_{t-1}$ as well. However, the solution of $G_{t-1}$ also represents
$\rho$ and as claimed in the first part of the Claim, the solution
does not use vertices $v^*_{1,\rho(1)}$, $\dots$,
$v^*_{k,\rho(k)}$. Therefore, no conflict arises between the solutions
of $G_t$ and $G_{t-1}$. 

Finally, we have to show that a solution for $\vec I$ implies that a
solution for $I$ exists. We say that a solution for $\vec I$ is {\em normal}
with respect to $G_t$ if the paths satisfying the demands in $G_t$ do
not leave $G_t$ (the vertices $v_{i,j}$, $v^*_{i,j}$ that were
identified are considered as part of both gadgets, so we allow the
paths to go through these vertices). We show by induction that the
solution for $\vec I$ is normal for every $G_t$. Suppose that this is true
for $G_{t-1}$. If some path $P$ satisfying a demand in $G_t$ leaves
$G_t$, then it has to enter either $G_{t-1}$ or $G_{t+1}$. If $P$
enters a vertex of $G_{t+1}$ that is not in $G_t$, then it cannot go
back to $G_t$: the only way to reach a vertex $v_{i,j}$ of $G_{t+1}$
is from vertex $a_i$, which has indegree 0. Therefore, let us suppose
that $P$ enters $G_{t-1}$ at some vertex $v^*_{i,j}$ of $G_{t-1}$. The
only way the path can return to $G_t$ is via some vertex $v^*_{i,j'}$
of $G_{t-1}$ with $j'\ge j$. By the induction hypothesis, the
solution is normal with respect to $G_{t-1}$, thus the second part
of the Claim implies that there is a unique $j$ such that $v^*_{i,j}$
is not used by the paths satisfying the demands in $G_{t-1}$. As $P$
can use only this vertex $v^*_{i,j}$, it follows that $j'=j$ and hence path $P$ does not use any vertex of
$G_{t-1}$ not in $G_t$, In other words, $P$ does not leave $G_t$.

Suppose now that the solution is normal with respect to every $G_t$,
which means that it induces a solution for every gadget. Suppose that the
solution of gadget $G_t$ represents mapping $\rho_t$. We claim that
every $\rho_t$ is the same. Indeed, if $\rho_t(i)=j$, then the
solution of $G_t$ uses vertex $v_{i,j}$ of $G_t$, which is identical
to vertex $v^*_{i,j}$ of $G_{t-1}$. This means that the solution of
$G_{t-1}$ does not use $v^*_{i,j}$, and by the second part of the Claim, this
is only possible if $\rho_{t-1}(i)=j$. Thus $\rho_{t-1}=\rho_t$ for
every $1<i \le m$, let $\rho$ be this mapping. Again by the claim, $\rho$ hits
every set $S_t$ in instance $I$, thus $\rho$ is a solution of $I$.
\end{proof}

For our main proof we will also need the following lemma. 
\begin{lemma}[\cite{Bienstock89-Gra}]
\label{lem:pathdivide}
Let $G$ be a graph (possibly with parallel edges) having pathwidth at
most $w$. Let $G'$ be obtained from $G$ by subdividing some of the
edges. Then the pathwidth of $G'$ is at most $w+1$. 
\end{lemma}

\begin{theorem}
\label{lem:disjpath}
Assuming the ETH, there is no $2^{o(w\log w)}\cdot n^{O(1)}$ time
algorithm  for \disjpath.
\end{theorem}

\begin{proof}
Let $\vec I$ be a instance of \dirpath\ on a directed graph $D$ having
pathwidth $w$. We transform $D$ into an undirected graph $G$, where two
adjacent vertices $v_\text{in}$, $v_{\text{out}}$ correspond to each
vertex $v$ of $D$, and if $\ora{uv}$ is an edge of $D$, then we
introduce a new vertex $e_{uv}$ that is adjacent to both $u_\text{out}$ and
$v_\text{in}$. It is not difficult to see that the pathwidth of $G$ is at
most $2w+2=O(w)$: $G$ can be obtained from the underlying graph of $D$
by duplicating vertices (which at most doubles the size of each bag) and
subdividing edges (which increases pathwidth at most by one).

Let $I$ be an instance of \disjpath\ on $G$ where there is a demand
$(v_\text{out},u_\text{in})$ corresponding to every demand of $(v,u)$
of $\vec I$. It is clear that if $\vec I$ has a solution, then $I$ has
a solution as well: every directed path from $u$ to $v$ in $D$ can be
turned into a path connecting $u_\text{out}$ and $v_\text{in}$ in $G$.
However, the converse is not true: it is possible that an undirected
path $P$ in $G$ reaches $v_\text{in}$ from $e_{uv}$ and instead of
continuing to $v_\text{out}$, it continues to some $e_{wv}$. In this
case, there is no directed path corresponding to $P$ in $D$. We add
further edges and demands to forbid such paths.

Let $B_1$, $\dots$, $B_n$ be a path decomposition of $G$ having width
$w'=O(w)$. For every vertex $x$ of $G$, let $\ell(x)$ and $r(x)$ be
the index of the first and last bags, respectively, were $x$
appears. It is well-known that the decomposition can be chosen such
that $r(x)\neq r(y)$ for any
two vertices $x$ and $y$.

We modify $G$ to obtain a graph $G'$ the following way. If vertex $v$
has $d$ inneighbors $u_1$, $\dots$, $u_d$ in $D$, then $v_\text{in}$
has $d+1$ neighbors in $G$: $v_\text{out}$ and $d$ vertices
$e_{u_1v}$, $\dots$, $e_{u_dv}$. Suppose that the neighbors of $v$ are
ordered such that $r(e_{u_1v})< \dots < r(e_{u_dv})$. We introduce
$2d-2$ new vertices $v^s_1$, $\dots$, $v^s_{d-1}$, $v^t_1$, $\dots$,
$v^t_{d-1}$ such that $v^s_i$ and $v^t_i$ are both adjacent to
$e_{u_iv}$ and $e_{u_{i+1}v}$.  For every $1\le i\le d-1$, we
introduce a new demand $(v^s_i,v^t_i)$. Repeating this procedure for
every vertex $v$ of $D$ creates an instance $I'$ of undirected
\disjpath\ on a graph $G'$.

We show that these new vertices and edges increase the pathwidth at
most by a constant factor. Observe that $G'$ can be obtained from $G$
by adding two parallel edges between $e_{u_iv}$ and $e_{u_{i+1}v}$ and
subdividing them. Thus by Lemma~\ref{lem:pathdivide}, all we need to
show is that adding these new edges increases pathwidth only by a
constant factor. If $r(e_{u_iv})\ge \ell(e_{u_{i+1}v})$, then the
parallel edges between $e_{u_iv}$ and $e_{u_{i+1}v}$ can be added
without changing the path decomposition: bag $B_{r(e_{iv})}$ contains
both vertices. If $r(e_{u_iv})< \ell(e_{u_{i+1}v})$, then let us
insert vertex $e_{u_iv}$ into every bag $B_j$ for $r(e_{u_iv})< j \le
\ell(e_{u_{i+1}v})$. Now bag $B_{\ell(e_{u_{i+1}v})}$ contains both
$e_{u_iv}$ and $e_{u_{i+1}v}$, thus we can add two parallel edges
between them. Note that vertex $v_\text{in}$ appears in every bag
where $e_{u_iv}$ is inserted: if not, then either $v_\text{in}$ does
not appear in bags with index at most $r(e_{u_iv})$, or it does not
appear in bags with index at least $\ell(e_{u_{i+1}v})$, contradicting
the fact that $v_\text{in}$ is adjacent to both $e_{u_iv}$ and
$e_{u_{i+1}v}$. Furthermore, vertices $e_{u_iv}$ and $e_{u_jv}$ are not
inserted into the same bag for any $i\neq j$: if $j>i$, then
$r(e_{u_jv})> r(e_{u_{i+1}v})\ge \ell(e_{u_{i+1}v})$.  Therefore,
the number of new vertices in each bag is at most the original size of the bag,
i.e., the size of each bag increases by at most a factor of 2.

We claim that $I'$ has a solution if
and only if $\vec I$ has. If $\vec I$ has a solution, then the
directed path satisfying demand $(u,v)$ gives in a natural way an
undirected path in $G'$ that satisfies demand
$(u_\text{out},v_\text{in})$. Thus we can obtain a pairwise disjoint
collection of paths that satisfy the demands of the form
$(u_\text{out},v_\text{in})$.  Note that if $v_\text{out}$, $e_{u_1v}$, $\dots$,
$e_{u_dv}$ are the neighbors of $v_\text{in}$ in $G'$, then the paths
in this collection use at most one of the vertices $e_{u_1v}$, $\dots$,
$e_{u_dv}$, say, $e_{u_jv}$. Now we can satisfy the demands
$(v^s_i,v^t_i)$ for every $1\le i \le d-1$: for $i<j$, we can use the
path $v^s_ie_{u_iv}v^t_i$, and for $i\ge j$, we can use the path
$v^s_ie_{u_{i+1}v}v^t_i$. Thus instance $I'$ has a solution.

For the other direction, suppose that $I'$ has a solution. Let us call
a path of this solution a {\em main path} if it satisfies a demand of
the form $(u_\text{out},v_\text{in})$. We claim that if $v_\text{in}$
is an internal vertex of a main path $P$, then $P$ contains
$v_\text{out}$ as well. Otherwise, $P$ has to contain at least two
of the neighbors $e_{u_1v}$, $\dots$, $e_{u_dv}$ of
$v_\text{out}$. In this case, less than $d-1$ vertices out of
$e_{u_1v}$, $\dots$, $e_{u_dv}$ remain available for the $d-1$ demands
$(v^s_1,v^t_1)$, $\dots$, $(v^s_d,v^t_d)$, a contradiction.

Consider a main path $P$ that satisfies a demand
$(u_\text{out},v_\text{in})$ of $I'$. Clearly, $P$ cannot go through
any terminal vertex other than $u_\text{out}$ and $v_\text{in}$. As
$u$ has indegree 0 in $D$, path $P$ has to go to some $e_{uw}$ and
then to $w_\text{in}$ after starting from $u_\text{out}$. By our claim
in the previous paragraph, the next vertex has to be $w_\text{out}$,
then again some $e_{wz}$ and $z_{\text{in}}$ and so on. Thus there is
a directed path in $D$ that corresponds to $P$ in $G'$. This means
that directed paths corresponding to the main paths of the solution
for $I'$ form a solution for $\vec I$.
\end{proof}

\section{Chromatic Number}
In this section, we give another lower bound result for a problem that is known to admit an algorithm with running time 
 $w^{O(w)}\cdot n^{O(1)}$ on graphs with treewidth $w$. In particular we show that the running time cannot be
improved to $2^{o(w\log w)}\cdot n^{O(1)}$ unless the ETH collapses. Given a graph $G$, a function $f:V(G)\rightarrow \{1,\ldots,\ell\}$, is called an {\em $\ell$-proper coloring} of $G$, if for any edge $uv\in E(G)$, we have that $f(u)\neq f(v)$. The chromatic number of a graph $G$ is the minimum positive integer $\ell$ for which $G$ admits an proper $\ell$-coloring and is denoted by $\chi(G)$. In the \chnum\ problem, we are given a graph $G$ and objective is to find the value of 
$\chi(G)$. It is well known that if $G$ has treewidth $w$ then $\chi(G)\leq w+1$. Using, this we can obtain an algorithm for 
\chnum\ running in time $w^{O(w)}\cdot n^{O(1)}$ on graphs with treewidth $w$ \cite{DBLP:journals/dam/JansenS97}. We show that in fact this running time is optimal. 

In what follows, we give a lower bound for a parameter even larger than the treewidth of the input graph. Given a graph 
$G$, a subset of vertices $C$ is called {\em vertex cover} if for every $uv\in E$, either $u\in C$ or $v\in C$. In other words, $G-C$ is an independent set. In particular, we will study the following parameterization of the problem. 

\problem{\chnum}
 {A graph $G$ together with a vertex cover $C$ of size at  most $k$ and a positive integer $\ell$.}{$k$}
{Is $\chi(G)\leq \ell$?}
It is well known that if $G$ has a vertex cover of size $k$, then its treewidth is upper bounded by $k+1$ and thus we can test whether $\chi(G)\leq \ell$ in time  $k^{O(k)}\cdot n^{O(1)}$. 


\begin{theorem}
\label{thm:chromaticnum}
Assuming the ETH, there is no $2^{o(k\log k)}\cdot n^{O(1)}$ time
algorithm  for \chnum\ parameterized by vertex cover number.
\end{theorem}
\begin{proof}
We prove the theorem by a reduction from the \kkpermcl\  problem.  Let $(I,k)$ be an instance of 
\kkpermcl\ consisting of a graph $H$ over the vertex set $[k]\times [k]$  and a positive integer $k$. Recall that in the \kkpermcl\  problem, the goal is to check whether there is a clique containing exactly one vertex from each {\em row,} and containing exactly one vertex from each {\em column.} In other words, the
vertices selected in the solution are $(1,\rho(1))$, $\dots$, $(k,\rho(k))$ for some
{\em permutation} $\rho$ of $[k]$.

Now we show how to construct the graph $G$, an input to  \chnum\ starting from $H$. The vertex set of $G$ consists of the following set of vertices and edges.  
\begin{itemize}
\item We have two cliques $C_a$ and $C_b$ of size $k$. The vertex set of $C_x$, $x\in \{a,b\}$, consists of 
$\{x_1,\ldots,x_k\}$. 
\item For every $i,j,x,y\in[k]$, $i\neq j$ and $x\neq y$, for which $(i,x)$ and $(j,y)$ are not adjacent in $H$, we have 
a new vertex $w_{xy}^{ij}$. We first make $w_{xy}^{ij}$ adjacent to $b_i$ and $b_j$. Finally, we add edges between 
 $w_{xy}^{ij}$ and $\{a_1,\ldots,a_k\}\setminus \{a_x,a_y\}$. 
\end{itemize}
This concludes the construction. 

We now show that $H$ has a permutation clique if and only if $\chi(G)= k$. Let the
vertices selected in the permutation clique are $(1,\rho(1))$, $\dots$, $(k,\rho(k))$ for some 
permutation $\rho$ of $[k]$. Now we define a proper $k$-coloring of $G$.  
For every $j\in[k]$, we color the vertex $a_j$ with $j$ and the vertex $b_j$  with $\rho(j)$. The only vertices that are left uncolored are $w_{xy}^{ij}$. Observe that the only colors that we can use for $w_{xy}^{ij}$ are 
$\{x,y\}$. Thus, if we can show that $Z=\{x,y\}\setminus \{\rho(i),\rho(j)\}$ is non-empty then we can use any color in $Z$ to color  $w_{xy}^{ij}$. But that follows since there is an edge between $(i,\rho(i))$ and $(j,\rho(j))$ and there is no edge between  $(i,x)$ and $(j,y)$ by definition of $w_{xy}^{ij}$.

Next we show the reverse direction. Let $f$ be a proper $k$-coloring function for $G$. Without loss of generality we can assume that  $f(a_j)=j$.  
For every 
$i\in[k]$, define $\rho(i)=f(b_i)$. Observe that since  $C_b$ is a clique and $f$ is a proper $k$-coloring for $H$ and hence in particular for $C_b$, we have that $\rho$ is a permutation of $[k]$. We claim that 
$(1,\rho(1))$, $\dots$, $(k,\rho(k))$ forms a permutation clique of $H$. Towards this we only need to show that there is an edge between every $(i,\rho(i))$ and  $(j,\rho(j))$ in $H$. For contradiction assume that $(i,\rho(i))$ and  $(j,\rho(j))$ are not adjacent in $H$. Consider, the vertex $w^{ij}_{\rho(i)\rho(j)}$ in $G$. It can only be colored with either $\rho(i)$ or $\rho(j)$. However, $f(b_i)=\rho(i)$ and $f(b_j)=\rho(j)$. This contradicts the fact that $f$ is a 
proper $k$-coloring of $G$. This concludes the proof in the reverse direction. 

Finally, observe that the vertices of $C_a$ and $C_b$ form a vertex cover for $G$ of size $2k$. 
The claim in the previous paragraph shows that an algorithm finding a $\chi(G)$ 
solves the instance $(I,k)$ of  \kkpermcl. Note the number of vertices in $G$ is bounded by a polynomial in $k$ and the vertex cover of $G$ is bounded by $2k$. Therefore a $2^{o(k \log k)}n^{O(1)}$ algorithm for \chnum\ would give a   $2^{o(k \log k)}$ algorithm for \kkpermcl, violating the ETH by Theorem~\ref{th:kkpermcl}. 
\end{proof}


\section{Conclusion}
In this paper we showed that several parameterized problems have
slightly superexponential running time unless the ETH fails.  In
particular we showed for four well-studied problems arising in three
different domains that the known superexponential algorithms are
optimal: assuming the ETH, there is no $2^{o(d\log d)}\cdot |I|^{O(1)}$ or
$2^{o(d\log |\Sigma|)}\cdot |I|^{O(1)}$ time algorithm for \clstr,
$2^{o(d\log d)}\cdot |I|^{O(1)}$ time algorithm for \dist, and
$2^{o(w\log w)}\cdot |I|^{O(1)}$ time algorithm for \disjpath\ and \chnum\
parameterized by treewidth. 
We believe that many further 
results of this form can be obtained by using the framework of the current paper. Two concrete problems 
that might be amenable to our framework are: 
\begin{itemize}
\setlength{\itemsep}{-5pt}
\item Are the known parameterized algorithms for  \textsc{Point Line Cover}~\cite{LangermanM05,GrantsonL06} and \textsc{Directed Feedback Vertex Set}~\cite{ChenLLOR08}, parameterized by the solution size, running in time $2^{O(k\log k)}\cdot  |I|^{O(1)}$ optimal?
\end{itemize}
In the conference version of this paper \cite{DBLP:conf/soda/LokshtanovMS11}, we asked further questions of this form, which have been answered by now.
\begin{itemize}
\item Is the $2^{O(k\log k)}\cdot  |I|^{O(1)}$ time parameterized algorithm for \textsc{Interval Completion}~\cite{VillangerHPT09} optimal?
In 2016, Cao \cite{DBLP:conf/soda/Cao16} showed that this is not the case: the problem can be solved in single-exponential time $6^k\cdot n^{O(1)}$. In fact, recently  Bliznets et al.~\cite{BliznetsFPP16} obtained an algorithm with running time $k^{O(\sqrt{k})}\cdot n^{O(1)}$ for \textsc{Interval Completion}. 

\item  Are the known parameterized algorithms for  \textsc{Hamiltonian Path}~\cite{grohe-flum-param},  \textsc{Connected Vertex Cover}~\cite{Moser05} and \textsc{Connected Dominating Set}~\cite{1070514}, parameterized by the treewidth $w$ of the input graph, running in time $2^{O(w\log w)}\cdot  |I|^{O(1)}$ optimal?
In 2011, Cygan et al.~introduced the technique of Cut \& Count \cite{Cygan2011}, which is able to give $2^{O(w)}\cdot |I|^{O(1)}$ time randomized algorithms for all these problems. Later, deterministic algorithms with this running time were  found \cite{BodlaenderCKN2015,FedorDS2014}. Cygan et al.~showed also that, assuming the ETH, there is no $2^{o(w\log w)}\cdot n^{O(1)}$ time algorithm for \textsc{Cycle Packing} on graphs of treewidth $w$.
\end{itemize}
It seems that our paper raised awareness in the field of parameterized
algorithms that tight lower bounds are possible even for running times
that may look somewhat unnatural, and in particular if a problem can
be solved in time $2^{O(k\log k)}\cdot n^{O(1)}$, then it is worth
exploring whether this can be improved to single-exponential or a
lower bound can be proved. The invention of the Cut \& Count technique and the related results of Cygan~et al.~\cite{Cygan2011} seem to be influenced by this realization. By now, there are other papers building on our work and investigating the optimality of $2^{O(k\log k)}\cdot n^{O(1)}$ time algorithms in the context of bounded-treewidth graphs or graph modification problems \cite{DBLP:conf/wg/BonnetBKM16,DBLP:journals/tcs/BroersmaGP13,DBLP:conf/mfcs/Pilipczuk11,DBLP:conf/isaac/DrangeDH14}



\bibliographystyle{siam}
\bibliography{superexp}

\begin{thebibliography}{10}

\bibitem{DBLP:conf/icalp/AlonLS09}
{\sc N.~Alon, D.~Lokshtanov, and S.~Saurabh}, {\em Fast {FAST}}, in Proceedings
  of the 36th International Colloquium, on Automata, Languages and Programming
  (ICALP), vol.~5555 of Lecture Notes in Computer Science, Springer, 2009,
  pp.~49--58.

\bibitem{MR1411787}
{\sc N.~Alon, R.~Yuster, and U.~Zwick}, {\em Color-coding}, J. Assoc. Comput.
  Mach., 42 (1995), pp.~844--856.

\bibitem{BadoiuCIS05}
{\sc M.~B{\u{a}}doiu, J.~Chuzhoy, P.~Indyk, and A.~Sidiropoulos}, {\em
  Low-distortion embeddings of general metrics into the line}, in Proceedings
  of the 37th Annual ACM Symposium on Theory of Computing (STOC), ACM, 2005,
  pp.~225--233.

\bibitem{BadoiuDGRRRS05}
{\sc M.~B{\u{a}}doiu, K.~Dhamdhere, A.~Gupta, Y.~Rabinovich, H.~R{\"a}cke,
  R.~Ravi, and A.~Sidiropoulos}, {\em Approximation algorithms for
  low-distortion embeddings into low-dimensional spaces}, in Proceedings of the
  16th Annual ACM-SIAM Symposium on Discrete Algorithms (SODA), ACM and SIAM,
  2005, pp.~119--128.

\bibitem{BadoiuIS07}
{\sc M.~B{\u{a}}doiu, P.~Indyk, and A.~Sidiropoulos}, {\em Approximation
  algorithms for embedding general metrics into trees}, in Proceedings of the
  18th Annual ACM-SIAM Symposium on Discrete Algorithms (SODA), ACM and SIAM,
  2007, pp.~512--521.

\bibitem{DBLP:journals/jair/BeckerBG00}
{\sc A.~Becker, R.~Bar-Yehuda, and D.~Geiger}, {\em Randomized algorithms for
  the loop cutset problem}, J. Artif. Intell. Res. (JAIR), 12 (2000),
  pp.~219--234.

\bibitem{Bienstock89-Gra}
{\sc D.~Bienstock}, {\em Graph searching, path-width, tree-width and related
  problems (a survey)}, in Reliability of computer and communication networks
  (New Brunswick, NJ, 1989), vol.~5 of DIMACS Ser. Discrete Math. Theoret.
  Comput. Sci., Amer. Math. Soc., Providence, RI, 1991, pp.~33--49.

\bibitem{BliznetsFPP16}
{\sc I.~Bliznets, F.~V. Fomin, M.~Pilipczuk, and M.~Pilipczuk}, {\em
  Subexponential parameterized algorithm for interval completion}, in
  Proceedings of the 27th Annual {ACM-SIAM} Symposium on Discrete Algorithms
  (SODA), {ACM and SIAM}, 2016, pp.~1116--1131.

\bibitem{BodlaenderCKN2015}
{\sc H.~L. Bodlaender, M.~Cygan, S.~Kratsch, and J.~Nederlof}, {\em
  Deterministic single exponential time algorithms for connectivity problems
  parameterized by treewidth}, Inform. and Comput., 243 (2015), pp.~86--111.

\bibitem{DBLP:conf/wg/BonnetBKM16}
{\sc {\'{E}}.~Bonnet, N.~Brettell, O.~Kwon, and D.~Marx}, {\em Parameterized
  vertex deletion problems for hereditary graph classes with a block property},
  in Proceedings of the 42nd International Workshop on Graph-Theoretic Concepts
  in Computer Science (WG), vol.~9941 of Lecture Notes in Computer Science,
  Springer, 2016, pp.~233--244.

\bibitem{DBLP:journals/tcs/BroersmaGP13}
{\sc H.~Broersma, P.~A. Golovach, and V.~Patel}, {\em Tight complexity bounds
  for {FPT} subgraph problems parameterized by the clique-width}, Theor.
  Comput. Sci., 485 (2013), pp.~69--84.

\bibitem{DBLP:journals/jcss/CaiJ03}
{\sc L.~Cai and D.~W. Juedes}, {\em On the existence of subexponential
  parameterized algorithms}, J. Comput. Syst. Sci., 67 (2003), pp.~789--807.

\bibitem{DBLP:conf/soda/Cao16}
{\sc Y.~Cao}, {\em Linear recognition of almost interval graphs}, in
  Proceedings of the 27th Annual {ACM-SIAM} Symposium on Discrete Algorithms
  (SODA), {ACM and SIAM}, 2016, pp.~1096--1115.

\bibitem{DBLP:conf/mfcs/ChenKX06}
{\sc J.~Chen, I.~A. Kanj, and G.~Xia}, {\em Improved upper bounds for vertex
  cover}, Theor. Comput. Sci., 411 (2010), pp.~3736--3756.

\bibitem{ChenLLOR08}
{\sc J.~Chen, Y.~Liu, S.~Lu, B.~O'Sullivan, and I.~Razgon}, {\em A
  fixed-parameter algorithm for the directed feedback vertex set problem}, J.
  {ACM}, 55 (2008), pp.~21:1--21:19.

\bibitem{chen-ma-wang-closest}
{\sc Z.~Chen, B.~Ma, and L.~Wang}, {\em A three-string approach to the closest
  string problem}, J. Comput. Syst. Sci., 78 (2012), pp.~164--178.

\bibitem{DBLP:journals/informs/CookS03}
{\sc W.~J. Cook and P.~D. Seymour}, {\em Tour merging via
  branch-decomposition}, {INFORMS} Journal on Computing, 15 (2003),
  pp.~233--248.

\bibitem{DBLP:books/sp/CyganFKLMPPS15}
{\sc M.~Cygan, F.~V. Fomin, L.~Kowalik, D.~Lokshtanov, D.~Marx, M.~Pilipczuk,
  M.~Pilipczuk, and S.~Saurabh}, {\em Parameterized Algorithms}, Springer,
  2015.

\bibitem{Cygan2011}
{\sc M.~Cygan, J.~Nederlof, M.~Pilipczuk, M.~Pilipczuk, J.~M.~M. van Rooij, and
  J.~O. Wojtaszczyk}, {\em Solving connectivity problems parameterized by
  treewidth in single exponential time}, in Proceedings of the 52nd Annual
  {IEEE} Symposium on Foundations of Computer Science (FOCS), 2011,
  pp.~150--159.

\bibitem{abs-1004-5012}
{\sc M.~Cygan and M.~Pilipczuk}, {\em Bandwidth and distortion revisited},
  Discrete Applied Mathematics, 160 (2012), pp.~494--504.

\bibitem{mr2352543}
{\sc F.~Dehne, M.~Fellows, M.~Langston, F.~Rosamond, and K.~Stevens}, {\em An
  {$O(2\sp {O(k)}n\sp 3)$} {FPT} algorithm for the undirected feedback vertex
  set problem}, Theory Comput. Syst., 41 (2007), pp.~479--492.

\bibitem{HMinorFree_JACM}
{\sc E.~D. Demaine, F.~V. Fomin, M.~Hajiaghayi, and D.~M. Thilikos}, {\em
  Subexponential parameterized algorithms on graphs of bounded-genus and
  {$H$}-minor-free graphs}, Journal of the ACM, 52 (2005), pp.~866--893.

\bibitem{BidimensionalSurvey_GD2004}
{\sc E.~D. Demaine and M.~Hajiaghayi}, {\em Fast algorithms for hard graph
  problems: Bidimensionality, minors, and local treewidth}, in Proceedings of
  the 12th International Symposium on Graph Drawing (GD), vol.~3383 of Lecture
  Notes in Computer Science, 2004, pp.~517--533.

\bibitem{1070514}
\leavevmode\vrule height 2pt depth -1.6pt width 23pt, {\em Bidimensionality:
  new connections between fpt algorithms and ptass}, in Proceedings of the 16th
  Annual ACM-SIAM Symposium on Discrete algorithms (SODA), {ACM and SIAM},
  2005, pp.~590--601.

\bibitem{MR2062503}
{\sc X.~Deng, G.~Li, Z.~Li, B.~Ma, and L.~Wang}, {\em A {PTAS} for
  distinguishing (sub)string selection}, in Proceedings of the 29th
  International Colloquium, on Automata, Languages and Programming (ICALP),
  vol.~2380 of Lecture Notes in Computer Science, Springer, 2002, pp.~740--751.

\bibitem{DBLP:conf/stacs/Dorn10}
{\sc F.~Dorn}, {\em Planar subgraph isomorphism revisited}, in Proceedings of
  the 27th International Symposium on Theoretical Aspects of Computer Science
  (STACS), vol.~5 of LIPIcs, Schloss Dagstuhl - Leibniz-Zentrum fuer
  Informatik, 2010, pp.~263--274.

\bibitem{DBLP:conf/esa/DornPBF05}
{\sc F.~Dorn, E.~Penninkx, H.~L. Bodlaender, and F.~V. Fomin}, {\em Efficient
  exact algorithms on planar graphs: Exploiting sphere cut decompositions},
  Algorithmica, 58 (2010), pp.~790--810.

\bibitem{downey-fellows-feasibility}
{\sc R.~G. Downey and M.~R. Fellows}, {\em Parameterized computational
  feasibility}, in Proceedings of the Second Cornell Workshop on Feasible
  Mathematics, P.~Clote and J.~Remmel, eds., Feasible Mathematics II,
  Birkhauser Boston, 1995, pp.~219--244.

\bibitem{MR2001b:68042}
{\sc R.~G. Downey and M.~R. Fellows}, {\em Parameterized Complexity},
  Monographs in Computer Science, Springer, New York, 1999.

\bibitem{DBLP:conf/isaac/DrangeDH14}
{\sc P.~G. Drange, M.~S. Dregi, and P.~van~'t Hof}, {\em On the computational
  complexity of vertex integrity and component order connectivity}, in
  Proceedings of the 25th International Symposium on Algorithms and Computation
  (ISAAC), vol.~8889 of Lecture Notes in Computer Science, Springer, 2014,
  pp.~285--297.

\bibitem{DBLP:journals/jgaa/Eppstein99}
{\sc D.~Eppstein}, {\em Subgraph isomorphism in planar graphs and related
  problems}, J. Graph Algorithms Appl., 3 (1999).

\bibitem{MR2000185}
{\sc P.~A. Evans, A.~D. Smith, and H.~T. Wareham}, {\em On the complexity of
  finding common approximate substrings}, Theoret. Comput. Sci., 306 (2003),
  pp.~407--430.

\bibitem{DBLP:conf/icalp/FellowsFLLRS09}
{\sc M.~R. Fellows, F.~V. Fomin, D.~Lokshtanov, E.~Losievskaja, F.~A. Rosamond,
  and S.~Saurabh}, {\em Distortion is fixed parameter tractable}, {TOCT}, 5
  (2013), pp.~16:1--16:20.

\bibitem{grohe-flum-param}
{\sc J.~Flum and M.~Grohe}, {\em Parameterized Complexity Theory}, Springer,
  Berlin, 2006.

\bibitem{DBLP:journals/jcss/FlumGW06}
{\sc J.~Flum, M.~Grohe, and M.~Weyer}, {\em Bounded fixed-parameter
  tractability and $\log^{2}$ nondeterministic bits}, J. Comput. Syst. Sci., 72
  (2006), pp.~34--71.

\bibitem{FedorDS2014}
{\sc F.~V. Fomin, D.~Lokshtanov, F.~Panolan, and S.~Saurabh}, {\em Efficient
  computation of representative families with applications in parameterized and
  exact algorithms}, J. {ACM}, 63 (2016), pp.~29:1--29:60.

\bibitem{FominLS09}
{\sc F.~V. Fomin, D.~Lokshtanov, and S.~Saurabh}, {\em An exact algorithm for
  minimum distortion embedding}, Theor. Comput. Sci., 412 (2011),
  pp.~3530--3536.

\bibitem{FredmanKS84}
{\sc M.~L. Fredman, J.~Koml{\'o}s, and E.~Szemer{\'e}di}, {\em Storing a sparse
  table with o(1) worst case access time}, J. ACM, 31 (1984), pp.~538--544.

\bibitem{MR2071152}
{\sc J.~Gramm, J.~Guo, and R.~Niedermeier}, {\em On exact and approximation
  algorithms for distinguishing substring selection}, in Proceedings of the
  14th International Symposium on Fundamentals of Computation Theory (FCT),
  vol.~2751 of Lecture Notes in Computer Science, Springer, 2003, pp.~195--209.

\bibitem{MR1984615}
{\sc J.~Gramm, R.~Niedermeier, and P.~Rossmanith}, {\em Fixed-parameter
  algorithms for closest string and related problems}, Algorithmica, 37 (2003),
  pp.~25--42.

\bibitem{GrantsonL06}
{\sc M.~Grantson and C.~Levcopoulos}, {\em Covering a set of points with a
  minimum number of lines}, in Proceedings of the 6th Italian Conference on
  Algorithms and Complexity (CIAC), vol.~3998 of Lecture Notes in Computer
  Science, 2006, pp.~6--17.

\bibitem{GGHNW-2005}
{\sc J.~Guo, J.~Gramm, F.~H{\"{u}}ffner, R.~Niedermeier, and S.~Wernicke}, {\em
  Compression-based fixed-parameter algorithms for feedback vertex set and edge
  bipartization}, J. Comput. Syst. Sci., 72 (2006), pp.~1386--1396.

\bibitem{GuptaNRS04}
{\sc A.~Gupta, I.~Newman, Y.~Rabinovich, and A.~Sinclair}, {\em Cuts, trees and
  $\ell_{\mbox{1}}$-embeddings of graphs}, Combinatorica, 24 (2004),
  pp.~233--269.

\bibitem{MR1894519}
{\sc R.~Impagliazzo, R.~Paturi, and F.~Zane}, {\em Which problems have strongly
  exponential complexity?}, J. Comput. System Sci., 63 (2001), pp.~512--530.

\bibitem{Indyk01}
{\sc P.~Indyk}, {\em Algorithmic applications of low-distortion geometric
  embeddings}, in Proceedings of the 42nd IEEE Symposium on Foundations of
  Computer Science (FOCS), IEEE, 2001, pp.~10--33.

\bibitem{DBLP:journals/dam/JansenS97}
{\sc K.~Jansen and P.~Scheffler}, {\em Generalized coloring for tree-like
  graphs}, Discrete Applied Mathematics, 75 (1997), pp.~135--155.

\bibitem{Karp75}
{\sc R.~M. Karp}, {\em On the computational complexity of combinatorial
  problems}, Netwroks, 5 (1975), pp.~45--68.

\bibitem{KenyonRS04}
{\sc C.~Kenyon, Y.~Rabani, and A.~Sinclair}, {\em Low distortion maps between
  point sets}, in Proceedings of the 36th Annual ACM Symposium on Theory of
  Computing (STOC), ACM, 2004, pp.~272--280.

\bibitem{MR1994748}
{\sc J.~K. Lanctot, M.~Li, B.~Ma, S.~Wang, and L.~Zhang}, {\em Distinguishing
  string selection problems}, Inform. and Comput., 185 (2003), pp.~41--55.

\bibitem{LangermanM05}
{\sc S.~Langerman and P.~Morin}, {\em Covering things with things}, Discrete
  {\&} Computational Geometry, 33 (2005), pp.~717--729.

\bibitem{506150}
{\sc M.~Li, B.~Ma, and L.~Wang}, {\em On the closest string and substring
  problems}, J. ACM, 49 (2002), pp.~157--171.

\bibitem{Linial02}
{\sc N.~Linial}, {\em Finite metric-spaces---combinatorics, geometry and
  algorithms}, in Proceedings of the International Congress of Mathematicians,
  Vol. III, Beijing, 2002, Higher Ed. Press, pp.~573--586.

\bibitem{DBLP:conf/soda/LokshtanovMS11}
{\sc D.~Lokshtanov, D.~Marx, and S.~Saurabh}, {\em Slightly superexponential
  parameterized problems}, in Proceedings of the 22nd Annual {ACM-SIAM}
  Symposium on Discrete Algorithms (SODA), {ACM-SIAM}, 2011, pp.~760--776.

\bibitem{Lynch75}
{\sc J.~F. Lynch}, {\em The equivalence of theorem proving and the
  interconnection problem}, ACM SIGDA Newsletter, 5 (1975), pp.~31--65.

\bibitem{DBLP:journals/siamcomp/MaS09}
{\sc B.~Ma and X.~Sun}, {\em More efficient algorithms for closest string and
  substring problems}, SIAM J. Comput., 39 (2009), pp.~1432--1443.

\bibitem{marx-closest-full}
{\sc D.~Marx}, {\em Closest substring problems with small distances}, SIAM
  Journal on Computing, 38 (2008), pp.~1382--1410.

\bibitem{MR87a:05097}
{\sc B.~Monien}, {\em How to find long paths efficiently}, in Analysis and
  design of algorithms for combinatorial problems (Udine, 1982), vol.~109 of
  North-Holland Math. Stud., North-Holland, Amsterdam, 1985, pp.~239--254.

\bibitem{Moser05}
{\sc H.~Moser}, {\em Exact algorithms for generalizations of vertex cover},
  Institut {f\"{u}r} Informatik, Friedrich-Schiller-Universit{\"{a}}t Jena,
  2005.

\bibitem{NaorSS95}
{\sc M.~Naor, L.~J. Schulman, and A.~Srinivasan}, {\em Splitters and
  near-optimal derandomization}, in Proceedings of the 36th IEEE Symposium on
  Foundations of Computer Science (FOCS), IEEE, 1995, pp.~182--191.

\bibitem{MR2223196}
{\sc R.~Niedermeier}, {\em Invitation to fixed-parameter algorithms}, vol.~31
  of Oxford Lecture Series in Mathematics and its Applications, Oxford
  University Press, Oxford, 2006.

\bibitem{DBLP:conf/mfcs/Pilipczuk11}
{\sc M.~Pilipczuk}, {\em Problems parameterized by treewidth tractable in
  single exponential time: {A} logical approach}, in Proceedings of the 36th
  International Symposium on Mathematical Foundations of Computer Science
  (MFCS), vol.~6907 of Lecture Notes in Computer Science, Springer, 2011,
  pp.~520--531.

\bibitem{ReedSmithVetta-OddCycle}
{\sc B.~Reed, K.~Smith, and A.~Vetta}, {\em Finding odd cycle transversals},
  Operations Research Letters, 32 (2004), pp.~299--301.

\bibitem{RobertsonS95b}
{\sc N.~Robertson and P.~D. Seymour}, {\em Graph minors {X}{I}{I}{I}. {T}he
  disjoint paths problem}, J. Comb. Theory, Ser. B, 63 (1995), pp.~65--110.

\bibitem{Scheffler94}
{\sc P.~Scheffler}, {\em A practical linear time algorithm for disjoint paths
  in graphs with bounded tree-width}, FU Berlin, Fachbereich 3 Mathematik,
  Tech. Rep. 396/1994 (1994).

\bibitem{VillangerHPT09}
{\sc Y.~Villanger, P.~Heggernes, C.~Paul, and J.~A. Telle}, {\em Interval
  completion is fixed parameter tractable}, SIAM J. Comput., 38 (2009),
  pp.~2007--2020.

\bibitem{DBLP:conf/faw/WangZ09}
{\sc L.~Wang and B.~Zhu}, {\em Efficient algorithms for the closest string and
  distinguishing string selection problems}, in Proceedings of the 3rd
  International Workshop on Frontiers in Algorithmics (FAW), vol.~5598 of
  Lecture Notes in Computer Science, Springer, 2009, pp.~261--270.

\bibitem{DBLP:conf/csr/Xiao08}
{\sc M.~Xiao}, {\em Simple and improved parameterized algorithms for
  multiterminal cuts}, Theory Comput. Syst., 46 (2010), pp.~723--736.

\end{thebibliography}
\end{document}